\documentclass[onecolumn, a4size, 11pt]{IEEEtran}
\usepackage{amsmath}
\usepackage{amssymb}
\usepackage{amsfonts}
\usepackage{graphicx}
\usepackage{epsfig}
\usepackage{subfigure}
\usepackage{psfrag}

\title{A Novel Mode Switching Scheme Utilizing Random Beamforming for Opportunistic Energy Harvesting
       \footnote{This paper has been presented in part at IEEE Wireless Communications and Networking Conference (WCNC), Shanghai, China, April 7-10, 2013.}
       \footnote{The authors are with the Department of Electrical and Computer Engineering, National University of Singapore (e-mail: elejhs@nus.edu.sg, elezhang@nus.edu.sg).}
       }

\author{Hyungsik Ju and Rui Zhang}

\begin{document}

\maketitle \thispagestyle{empty}

\begin{abstract}
Since radio signals carry both energy and information at the same time, a unified study on simultaneous wireless information and power transfer (SWIPT) has recently drawn a significant attention for achieving wireless powered communication networks. In this paper, we study a multiple-input single-output (MISO) multicast SWIPT network with one multi-antenna transmitter sending common information to multiple single-antenna receivers simultaneously along with opportunistic wireless energy harvesting at each receiver. From the practical consideration, we assume that the channel state information (CSI) is only known at each respective receiver but is unavailable at the transmitter. We propose a novel receiver mode switching scheme for SWIPT based on a new application of the conventional random beamforming technique at the multi-antenna transmitter, which generates artificial channel fading to enable more efficient energy harvesting at each receiver when the received power exceeds a certain threshold. For the proposed scheme, we investigate the achievable information rate, harvested average power and power outage probability, as well as their various trade-offs in quasi-static fading channels. Compared to a reference scheme of periodic receiver mode switching without random transmit beamforming, the proposed scheme is shown to be able to achieve better rate-energy trade-offs when the harvested power target is sufficiently large. Particularly, it is revealed that employing one single random beam for the proposed scheme is asymptotically optimal as the transmit power increases to infinity, and also performs the best with finite transmit power for the high harvested power regime of most practical interests, thus leading to an appealing low-complexity implementation. Finally, we compare the rate-energy performances of the proposed scheme with different random beam designs.
\end{abstract}


\setlength{\baselineskip}{1.3\baselineskip}
\newtheorem{definition}{\underline{Definition}}[section]
\newtheorem{fact}{Fact}
\newtheorem{assumption}{Assumption}
\newtheorem{theorem}{\underline{Theorem}}[section]
\newtheorem{lemma}{\underline{Lemma}}[section]
\newtheorem{corollary}{\underline{Corollary}}[section]
\newtheorem{proposition}{\underline{Proposition}}[section]
\newtheorem{example}{\underline{Example}}[section]
\newtheorem{remark}{\underline{Remark}}[section]
\newtheorem{algorithm}{\underline{Algorithm}}[section]
\newcommand{\mv}[1]{\mbox{\boldmath{$ #1 $}}}

\section{Introduction}\label{Intorduction}
Conventionally, fixed energy supplies (e.g. batteries) are employed to power energy-constrained wireless networks, such as sensor networks. The lifetime of the network is typically limited, and is thus one of the most important considerations for designing such networks. To prolong the network's operation time, energy harvesting has recently attracted a great deal of attention since it enables scavenging energy from the environment and potentially provides unlimited power supplies for wireless networks.

Among other commonly used energy sources (e.g. solar and wind), radio signals radiated by ambient transmitters have drawn an upsurge of interest as a viable new source for wireless energy harvesting. Harvesting energy from radio signals has already been successfully implemented in applications such as passive radio-frequency identification (RFID) systems and body sensor networks (BSNs) for medical implants. More interestingly, wireless energy harvesting opens an avenue for the joint investigation of simultaneous wireless information and power transfer (SWIPT) since radio signals carry energy and information at the same time. SWIPT has recently been investigated for various wireless channels, e.g., the point-to-point additive white Gaussian noise (AWGN) channel \cite{Zhou}, the fading AWGN channel \cite{Liu}-\cite{Caspers}, the multi-antenna channel \cite{Zhang}-\cite{Park}, the relay channel \cite{Gurakan}, \cite{Narir}, and the multi-carrier based broadcast channel \cite{Ng}-\cite{Zhou2}.

To achieve maximal wireless energy transfer (WET) and wireless information transfer (WIT) simultaneously, one key challenge is to develop efficient and pragmatic receiver architectures to enable information decoding (ID) and energy harvesting (EH) from the same received signal at the same time \cite{Zhou}, \cite{Caspers}. Practically, two suboptimal receiver designs for SWIPT have been proposed in \cite{Zhang} based on the principle of orthogonalizing ID and EH processes, namely \emph{power splitting} and \emph{time switching}. The power splitting scheme splits the received signal into two streams of different power for ID and EH separately, while the time switching scheme switches the receiver between an ID mode and an EH mode from time to time. The optimal switching rules between ID versus EH modes for a point-to-point single-antenna fading channel subject to the co-channel interference have been derived in \cite{Liu} to maximize/minimize the information transmission rate/outage probability given an average harvested energy target. It was shown in \cite{Liu} that the time-fluctuation or fading of wireless channels is indeed beneficial for receiver mode-switching (time-switching) based SWIPT systems, where an ``opportunistic'' energy harvesting scheme is proved to be optimal, i.e., the receiver should switch to the EH mode when the channel power is larger than a certain threshold, and to the ID mode otherwise. Intuitively, this phenomenon can be explained as follows. Note that the received energy (in Joule) and amount of information (in bits) both scale linearly with time, but linearly and sub-linearly (logarithmically) with power, respectively; as a result, given the same signal energy for EH at receiver, it is desirable to have more significant power fluctuations such that a given target energy can be harvested during shorter peak-power periods, thus resulting in more time for receiving a higher amount of information (with the same energy left for ID).

In this paper, we further investigate the time-switching based SWIPT system in a multicast scenario, where one multi-antenna transmitter (Tx) broadcasts both energy and common information to multiple single-antenna receivers (Rxs) simultaneously over quasi-static multiple-input single-output (MISO) flat-fading channels, as shown in Fig. \ref{Fig_MulticastSWIPT}. We assume that Tx has an unlimited energy supply that provides constant transmit power while all Rxs have only limited energy sources  (e.g., rechargeable batteries) and thus need to replenish energy from the signals broadcast by Tx. Each Rx harvests energy and decodes information from the received signal via time switching, i.e., it can either decode information \emph{or} harvest energy from the received signal at any time, but \emph{not both}. It is worth noting that the number of Rxs in the network can be arbitrarily large, and thus it may not be practically feasible for Tx to gather the instantaneous channel state information (CSI) from all Rxs via dedicated feedback since this will increase the system complexity and overhead drastically with the increasing number of Rxs. Therefore, in this paper we consider a practical setup where the MISO channels from Tx to different Rxs are only known at each respective Rx but unavailable at Tx.

\begin{figure}
   \centering
   \includegraphics[width=0.45\columnwidth]{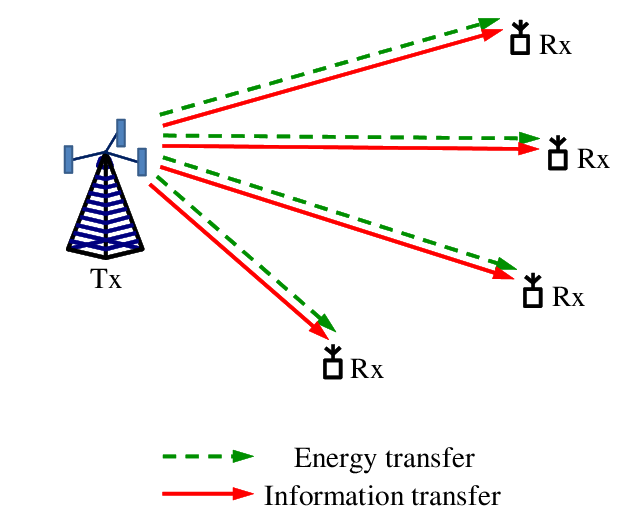}
   \caption{A MISO multicast network for SWIPT.}
   \label{Fig_MulticastSWIPT}
\end{figure}

In order to optimize the rate-energy (R-E) trade-offs achievable at each Rx, inspired by the result on the beneficial time-variation of fading channels for time-switching based SWIPT systems \cite{Liu}, in this paper we propose a new application of the celebrated ``random beamforming'' technique at the multi-antenna transmitter to generate artificial channel variations at each receiver to opportunistically harvest energy when the channel power exceeds a given threshold and decode information otherwise. This is realized by partitioning each transmission block with constant user channels into sub-blocks with equal duration in which independent random beams (RBs) are applied to generate artificial channel fading. Note that the use of random beamforming in this paper is motivated differently from that in the conventional setup for broadcasting with WIT only, which aims at achieving asymptotically interference-free independent information transmissions to multiple receivers in multi-antenna broadcast channels by exploiting multi-user diversity based partial channel feedback and transmission scheduling as the number of receivers increases to infinity \cite{Viswanath}, \cite{Sharif}. In contrast, for multicast SWIPT systems under our investigation, random beamforming is employed for generating artificial time-variation of channels to achieve better R-E trade-offs with time-switching receivers.

The main results of this paper are summarized as follows:

\begin{itemize}
   \item We propose a novel design with transmitter random beamforming and receiver time switching for MISO multicast SWIPT systems. We first characterize the performance trade-offs between WET and WIT by investigating the achievable rate and harvested power pair in a given transmission block with constant MISO AWGN channels, assuming Gaussian distributed random beams. Furthermore, we compare the R-E performance of our proposed scheme with that of a reference scheme with receiver periodic switching between ID and EH modes, but without random beamforming applied at Tx.

   \item We then extend our analysis for the MISO AWGN channel to MISO Rayleigh fading channel. We investigate the achievable average information rate and average harvested power at each Rx, and characterize their asymptotic trade-offs when the transmit power goes to infinity. It is shown that employing one single random beam for the proposed scheme achieves the best R-E trade-off asymptotically and also outperforms that of periodic switching.

   \item When Rx consumes significant amount of power at each block and/or the capacity of its energy storage device is limited, it may suffer from power shortage unless the amount of harvested power in each block is larger than a certain requirement. We thus study the ``power outage probability'' of the proposed scheme in fading MISO channels, which is also compared to that of the periodic switching in both asymptotic and finite transmit power regimes.

   \item In practice, transmit power is preferably to be constant for the maximal operation efficiency of transmitter amplifiers. However, the use of Gaussian distributed random beams for the proposed scheme can cause large transmit power fluctuations. We thus propose alternative random beam designs with constant transmit power, for which the R-E performance is characterized and compared with the case of Gaussian random beam.
\end{itemize}

The rest of this paper is organized as follows. Section \ref{Sec:SystemModel} introduces the proposed scheme as well as the reference scheme of periodic switching, and compare their harvested power and achievable information rate for one single block with the AWGN MISO channel. Section \ref{Sec:PerformanceAnalysis} investigates the R-E performances of the proposed and reference schemes in Rayleigh fading MISO channels. Section \ref{Sec:OtherBeams} compares the performances of the proposed scheme with different random beam designs. Finally, Section \ref{Sec:Conclusion} concludes the paper.

\emph{Notations:} In this paper, matrices and vectors are denoted by bold-face upper-case letters and lower-case letters, respectively. ${{\bf{I}}_N}$ denotes an $N \times N$ identity matrix and ${{\bf{0}}}$ represents a matrix with all zero entries. The distribution of a circularly symmetric complex Gaussian (CSCG) random vector with mean vector $\boldsymbol{\mu}$ and covariance matrix ${\boldsymbol{\Sigma }}$ is denoted by ${\mathcal{CN}}( {\boldsymbol{\mu} ,{\boldsymbol{\Sigma }}})$, and $ \sim $ stands for ``distributed as". ${{\mathbb{C}}^{a \times b}}$ and ${{\mathbb{R}}^{a \times b}}$ denote the spaces of $a \times b$ matrices with complex and real entries, respectively. $\left\| {\bf{z}} \right\|$ denotes the Euclidean norm of a complex vector $\bf{z}$. ${\mathbb{E}}\left[  \cdot  \right]$ represents the statistical expectation.

\section{System Model}\label{Sec:SystemModel}
As shown in Fig. \ref{Fig_MulticastSWIPT}, we consider a MISO multicast SWIPT system consisting of one Tx and multiple Rxs, e.g., sensors. Since Tx broadcasts a common signal to all Rxs, in this paper we focus on one particular Tx-Rx pair as shown in Fig. \ref{Fig_SystemModel} for the purpose of exposition, while the effect of multiuser channels on the performance of the considered system will be evaluated by simulation in Section \ref{Sec:PerformanceAnalysis}. We assume that Tx is equipped with $N_t > 1$ antennas and Rx is equipped with one single antenna. It is also assumed that the MISO channel from Tx to Rx follows quasi-static flat-fading, where the channel remains constant during each block transmission time, denoted by $T$, but varies from one block to another. It is further assumed that the channel in each block is perfectly known at Rx, but unknown at Tx.

\begin{figure}
   \centering
   \includegraphics[width=0.5\columnwidth]{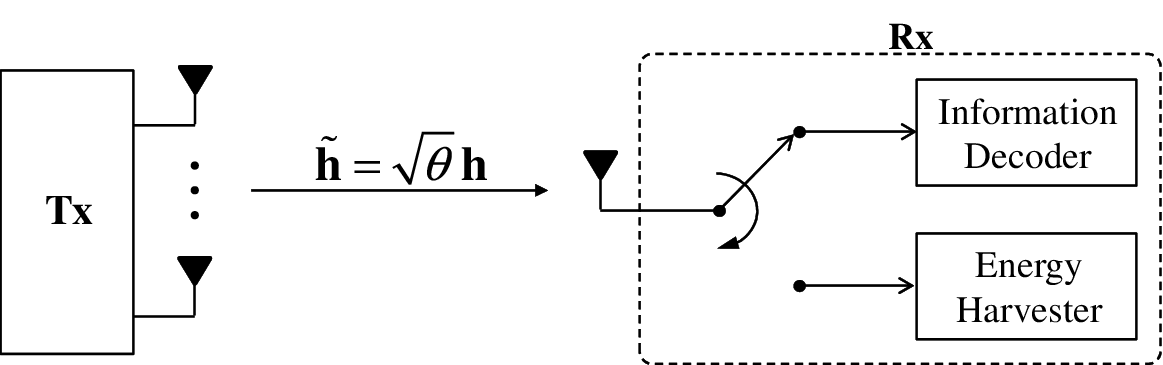}
   \caption{A MISO wireless system for SWIPT via receiver mode switching.}
   \label{Fig_SystemModel}
\end{figure}

The transmitted signal at the $i$th symbol interval in the $t\,$th transmission block is denoted by ${\bf{x}}_t \left( i \right) \in {{\mathbb{C}}^{{N_t} \times 1}}$. The covariance matrix of the transmitted signal is thus given by ${{\bf{S}}_{t,\,\bf{x}}} = \mathbb E [ {{\bf{x}}_t\left( i \right){\bf{x}}_t^{H}{{\left( i \right)}}} ] = \frac{P}{{{N_t}}}{{\bf{I}}_{{N_t}}}$, where $P$ denotes the constant transmit power, which is assumed to be equally allocated among $N_t$ transmit antennas. In addition, the MISO channel from Tx to Rx in the $t\,$th transmission block is denoted by ${\tilde{\bf{h}}}_t \in {{\mathbb{C}}^{{N_t} \times 1}}$, which is constant during each block. Without loss of generality, the MISO channel ${\bf{\tilde h}}_t$ can be modeled as ${\bf{\tilde h}}_t = \sqrt{\theta}\,  {\bf{h}}_t$, where $\theta$ and ${\bf{h}}_t \in \mathbb C^{N_t \times 1}$ denote the signal power due to distance-dependent attenuation and large-scale channel fading (assumed to be constant over all $t$'s for the time being) and the MISO channel due to small-scale channel fading in the $t\,$th block, respectively. The received signal at Rx is then expressed as
\begin{equation}\label{Eq_ReceivedSignal_General}
   {\begin{array}{l}
   {y_t}\left( i \right) = {\bf{\tilde h}}_t^T{{\bf{x}}_t}\left( i \right) + {z_t}\left( i \right) \\
   \,\,\,\,\,\,\,\,\,\,\,\,\,\,\, = \sqrt{\theta} {\bf{h}}_t^T{{\bf{x}}_t}\left( i \right) + {z_t}\left( i \right), \\
 \end{array}}
\end{equation}
where ${y}_t\left( i \right)$ and ${z}_t\left( i \right)$ denote the received signal and noise at Rx, respectively; it is assumed that ${z}_t\left( i \right) \sim {\mathcal{CN}}\left( {0,\sigma^2} \right)$, which is independent over both $t$ and $i$. In addition, since we can consider one block of interest without loss of generality, the block index $t$ will be omitted in the sequel for notational brevity.

In each block, Tx aims at achieving SWIPT to Rx. It is assumed that Rx is equipped with a rechargeable battery to store the energy harvested from the received signal, which is used to provide power to its operating circuits. Specifically, Rx harvests energy from the received signals when it is in the EH mode, while it decodes information in the ID mode. We assume that Rx switches between ID mode and EH mode as in \cite{Liu} and \cite{Zhang} since it is difficult yet to use the received signal for both ID and EH at the same time due to practical circuit limitations \cite{Zhou}. As in \cite{Liu}, ID mode and EH mode are represented by defining an indicator function as
\begin{equation}\label{Eq_ModeSelection}
   {\rho  = \left\{ {\begin{array}{*{20}{c}}
   {1,}  \\
   {0,}  \\
   \end{array}\begin{array}{*{20}{c}}
   {\,\,\,{\rm{ID}}\,\,{\rm{mode}}\,\,{\rm{is}}\,\,{\rm{active}}}  \\
   {\,\,\,{\rm{EH}}\,\,{\rm{mode}}\,\,{\rm{is}}\,\,{\rm{active.}}}  \\
   \end{array}} \right.}
\end{equation}

We consider two time switching schemes, namely ``\emph{periodic switching} (PS)'' and ``\emph{threshold switching} (TS)'' as elaborated next.

   \subsection{Reference Scheme: Periodic Switching}\label{Ref_OOS}

   \begin{figure}
      \centering
      \includegraphics[width=0.52\columnwidth]{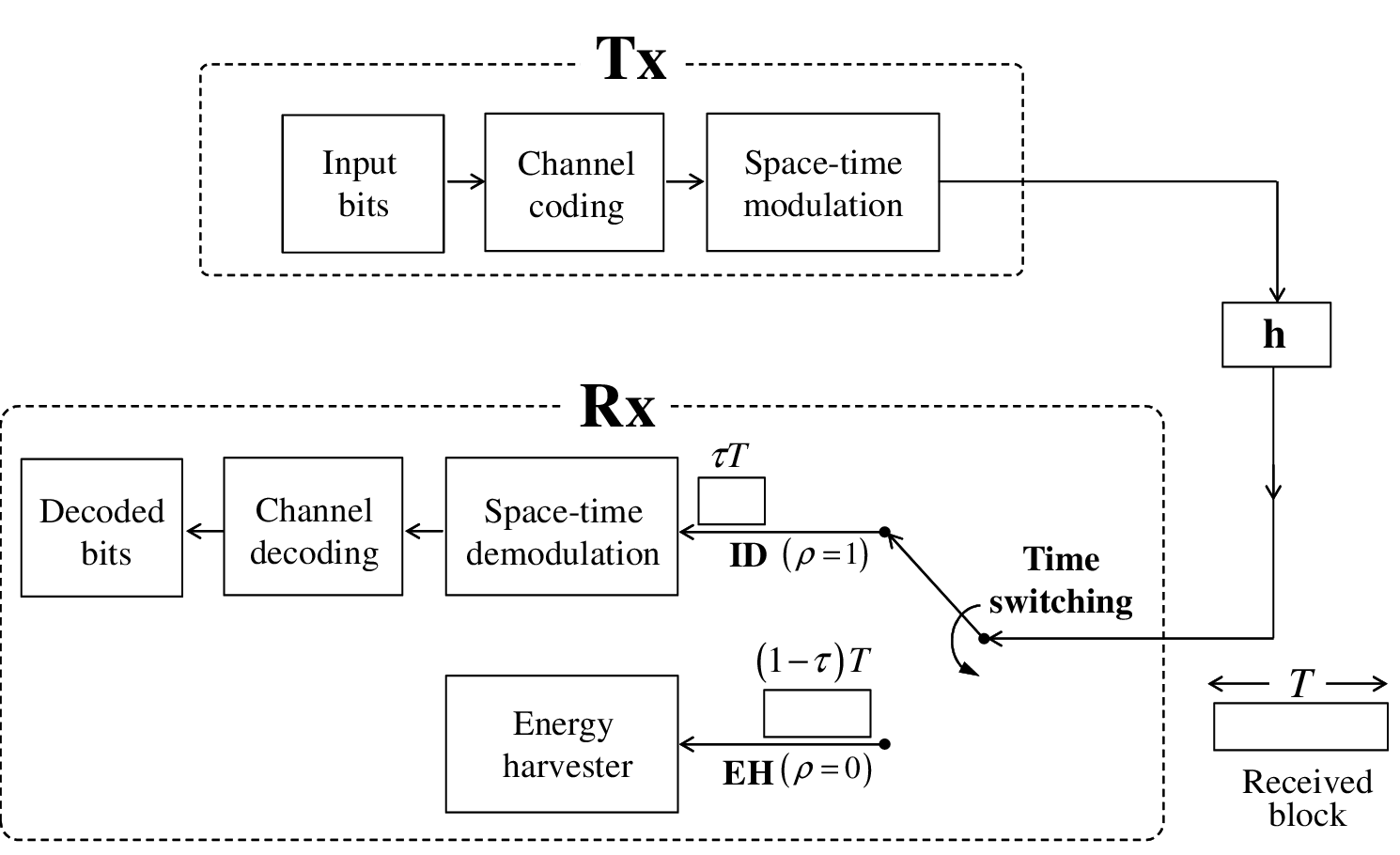}
      \caption{Transmitter and receiver structures for periodic switching (PS).}
      \label{Fig_OOS_SystemModel}
   \end{figure}

   As shown in Fig. \ref{Fig_OOS_SystemModel}, with PS, Rx sets $\rho  = 1$ during the first $\tau T$ amount of time in each transmission block, with $0 \le \tau \le 1$, and $\rho  = 0$ for the remaining block duration $(1-\tau)T$.\footnote{Ideally, with a given time allocation $\tau$, setting $\rho = 1$ or $0$ at the beginning of each block will not change the system performance; however, setting $\rho = 1$ initially is practically more favorable for Rx to implement block-wise time synchronization.} For given ${\bf{h}}$ and $\tau$, the amount of harvested energy normalized by $T$, i.e., \emph{average harvested power}, in a transmission block can be derived using ${{\bf{S}}_{\bf{x}}}$ as
   \begin{equation}\label{Eq_Energy_TimeSharing}
      {\begin{array}{l}
      {Q^{(\rm{P})}}\left( {H,\tau } \right) = \left( {1 - \tau } \right)\zeta  {\mathbb{E}}\left[ {{{\left\| {\sqrt{\theta} \, {{\bf{h}}^T}{\bf{x}}\left( i \right)} \right\|}^2}} \right] \\
      \,\,\,\,\,\,\,\,\,\,\,\,\,\,\,\,\,\,\,\,\,\,\,\,\,\,\,\,\,\,\, = \left( {1 - \tau } \right)\zeta \theta PH, \\
      \end{array}}
   \end{equation}
   where $H = \frac{1}{{{N_t}}}\left\| {\bf{h}} \right\|^2$ is the normalized average channel power, and $0 < \zeta \le 1$ is a constant reflecting the loss in the energy transducer when the harvested energy is converted to electrical energy to be stored. In (\ref{Eq_Energy_TimeSharing}), it has been assumed that the power harvested due to the receiver noise is negligible and thus is ignored. It is further assumed that $\zeta = 1$ in the sequel for notational brevity.

   The structure of Tx for PS is also shown in Fig. \ref{Fig_OOS_SystemModel}. Note that with PS, Rx can adjust $\tau$ based on its energy and rate requirements, as well as the channel condition. Since Tx keeps sending information symbols while Rx determines $\tau$ for switching between ID and EH modes based on its own channel quality, Rx observes an erasure AWGN channel and thus the erasure code \cite{Alon} should be employed at Tx for channel coding.\footnote{This is especially useful for the multicast network, where receivers can set different values of $\tau$ for decoding common information sent by the transmitter, based on their individual channel conditions and energy requirements.} The bit stream to be transmitted during a transmission block is thus first encoded by an erasure code. Space-time (ST) code is then applied to modulate the output bits from the erasure-code encoder, and the modulated symbols are transmitted by $N_t$ antennas. We consider a ST code of length $L$, denoted by matrix ${\bf{X}}^{(\rm{P})} \in {{\mathbb{C}}^{L \times {N_t}}}$. It is assumed that ${\bf{X}}^{(\rm{P})}$ is a capacity-achieving ST code.\footnote{Alamouti code \cite{Alamouti} is known as the capacity-achieving ST code when $N_t = 2$. For $N_t >2$, capacity-achieving ST code has not yet been found in general. In this paper, however, capacity-achieving ST code is assumed even when $N_t > 2$ to provide a performance upper bound for the system under consideration.} Tx transmits a sequence of ${\bf{X}}^{(\rm{P})}$'s in each transmission block. Considering ${\bf{X}}^{(\rm{P})}$ with $L$ consecutive transmitted symbols from each antenna, (\ref{Eq_ReceivedSignal_General}) is modified as
   \begin{equation}\label{Eq_ReceivedSignal_OOS}
      {{\bf{y}} = \sqrt{\theta}\,{{\bf{X}}^{({\rm{P}})}}{\bf{h}} + {\bf{z}},}
   \end{equation}
   where ${\bf{y}} \in {{\mathbb{C}}^{L \times 1}}$ and ${\bf{z}} \in {{\mathbb{C}}^{L \times 1}}$ denote the received signal vector and noise vector, respectively, and ${\bf{z}} \sim { \mathcal{CN}}\left( {{\bf{0}},\sigma^2 {{\bf{I}}_L}} \right)$. Since ${\bf{X}}^{(\rm{P})}$ is assumed to be a capacity-achieving ST code, the achievable rate of the channel in (\ref{Eq_ReceivedSignal_OOS}) can be shown equivalent to that of a MISO channel $\tilde{\bf{h}} = \sqrt{\theta}\,{\bf{h}}$ with input covariance matrix ${{\bf{S}}_{\bf{x}}} = \frac{P}{{{N_t}}}{{\bf{I}}_{{N_t}}}$. Assume that the number of ST coded blocks transmitted in each block is sufficiently large such that $\tau T$ is approximately an integer number of the ST block durations for any value of $\tau$. For given ${\bf{h}}$ and $\tau$, the information rate for PS can thus be expressed as
   \begin{equation}\label{Eq_Rate_TimeSharing}
      {{R^{(\rm{P})}}\left( H, \tau \right) = \tau{\log _2}\left( {1 + \frac{\theta P H}{\sigma^2}} \right),}
   \end{equation}
   Note that ${R^{(\rm{P})}}\left( H, \tau \right)$ is achievable when $N_t \le 2$, but is in general an upper bound on the achievable rate when $N_t > 2$ for given $\bf{h}$ and $\tau$.

   \subsection{Proposed Scheme: Threshold Switching}\label{Prof_TBS}
   \begin{figure}
      \centering
      \includegraphics[width=0.52\columnwidth]{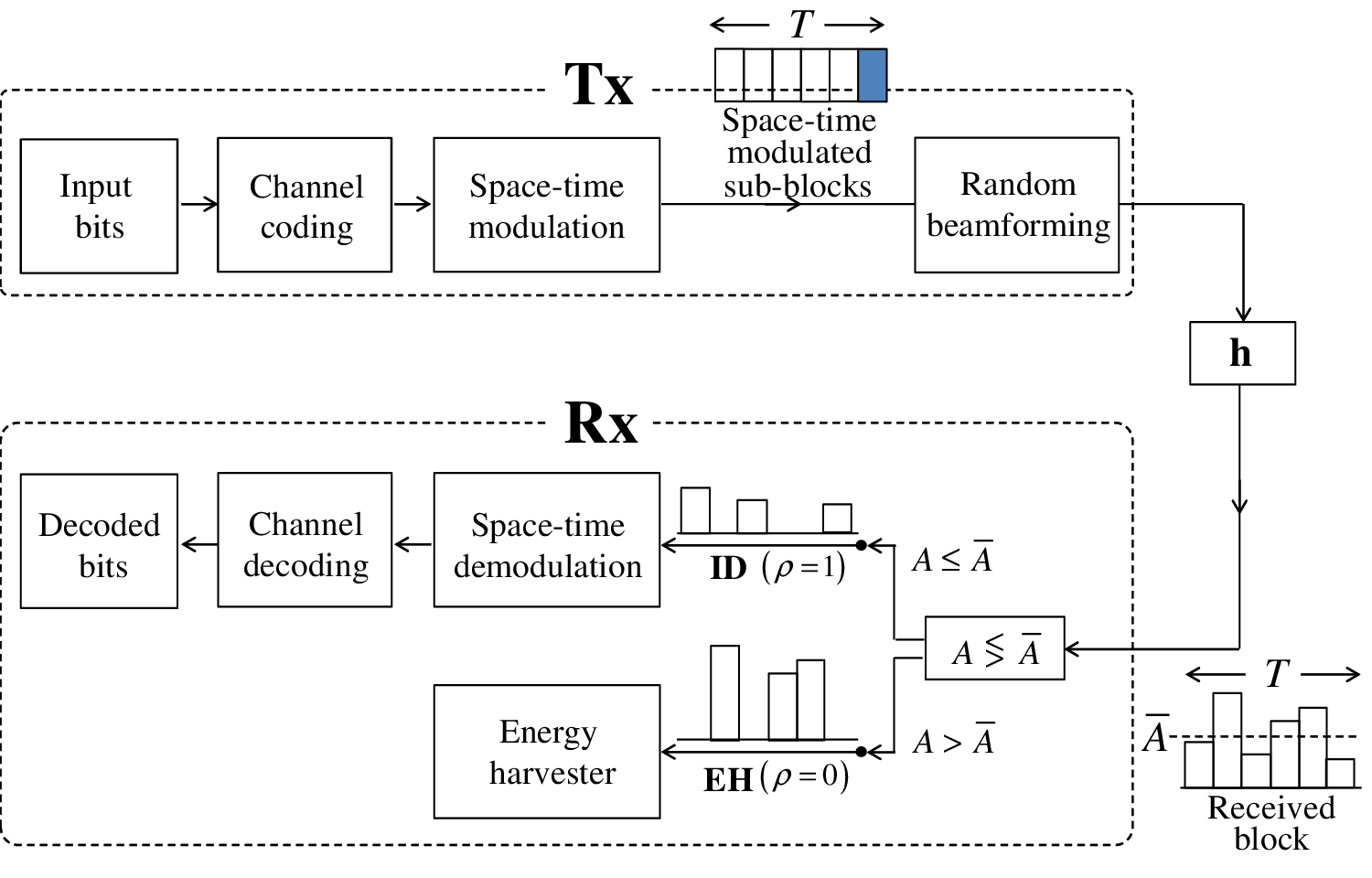}
      \caption{Transmitter and receiver structures for threshold switching (TS).}
      \label{Fig_TBS_SystemModel}
   \end{figure}

   As shown in Fig. \ref{Fig_TBS_SystemModel}, the TS scheme is designed to take advantage of the received signal power fluctuations induced by transmit random beamforming within each transmission block for opportunistic EH/ID mode switching, even with a constant MISO channel $\bf{h}$. For this purpose, each transmission block is further divided into $K$ sub-blocks each consisting of one or more ST codewords, and artificial channel fading over different sub-blocks is generated by multi-antenna random beamforming at Tx.

   Furthermore, at the $k$th sub-block, $k = 1, \cdots, K$, Rx determines whether to switch to ID mode or EH mode based on $A\left( k \right)$, which denotes the channel power at the $k$th sub-block normalized by $\theta$ and $P$ (to be specified later). According to \cite{Liu}, in the presence of received channel power fluctuations, the optimal mode switching rule that achieves the optimal trade-off between the maximum harvested energy and information rate in a transmission block is given by
   \begin{equation}\label{Eq_Opt_1_Solution}
      {\rho \left( k \right) = \left\{ {\begin{array}{*{20}{c}}
      {1,}  \\
      {0,}  \\
      \end{array}\begin{array}{*{20}{c}}
      {\,\,\,{\rm{if}}\,\,A\left( k \right) \le \bar A}  \\
      {\,\,\,{\rm{otherwise}},}  \\
      \end{array}} \right.}
   \end{equation}
   where $\bar A \ge 0$ is a pre-designed threshold on the normalized channel power $A(k)$. It is noted that choosing EH or ID mode at the $k$th sub-block is determined by the normalized channel power $A\left( k \right)$ as compared to the threshold $\bar A$, or equivalently the received signal power $\theta PA\left(k\right)$ as compared to the threshold $\theta P \bar A$; thus, ID mode is selected, i.e., $\rho \left( k \right) = 1$, if the received signal power is no greater than $\theta P \bar {A}$ and EH mode is selected, i.e., $\rho \left( k \right) = 0$, otherwise.

   Artificial channel fading over sub-blocks is generated at Tx by using $N$ RBs simultaneously, $1 \le N \le N_t$. Denote the $n$th RB at the $k$th sub-block as ${\boldsymbol{\phi} _n}\left( k \right) \in {{\mathbb C}^{{N_t} \times 1}}$, where ${\mathbb{E}} [ {{\boldsymbol{\phi} _n}\left( k \right){\boldsymbol{\phi} _{n}^{H}}{{\left( {k} \right)}}} ] = \frac{1}{N_t}{{\bf{I}}_{{N_t}}}$ and ${\mathbb{E}} [ {{\boldsymbol{\phi} _n}\left( k \right){\boldsymbol{\phi} _{m}^{H}}{{\left( {j} \right)}}} ] = {{\bf{0}}}$ if $k \ne j$ and/or $n \ne m$. Then it follows that $A\left( k \right) = \frac{1}{N}{\left\| {{{\bf{a}}}\left( k \right)} \right\|^2}$, where ${\bf{a}}(k) = {\bf{\Phi}}^T {\left( k \right)}{\bf{h}} \in \mathbb C^{N \times 1}$ is the equivalent MISO channel at the $k$th sub-block generated by ${\bf{\Phi}} \left( k \right) = [ {{\boldsymbol{\phi} _1}( k )\,\,{\boldsymbol{\phi} _2}( k )\,\, \cdots \,\,{\boldsymbol{\phi} _N}( k )} ]$, which is assumed to be a pre-designed pseudo random sequence and known to all Rxs.\footnote{Each Rx can estimate ${\bf{a}}\left( k \right)$'s without knowledge of ${\bf{\Phi}} \left( k \right)$'s by employing conventional channel estimation over all sub-blocks. However, such an implementation incurs high training overhead. When ${\bf{\Phi}} \left( k \right)$'s are assumed to be known at all Rxs, however, each Rx only needs to estimate ${\bf{h}}$ at the beginning of each block to obtain ${\bf{a}}(k)$'s and thus the overhead for channel estimation can be significantly reduced.}

   Similarly to PS, the erasure code should be employed in the case of TS for channel coding since the set of sub-blocks used for ID according to (\ref{Eq_Opt_1_Solution}) are in general randomly distributed within a transmission block with $\bar A > 0$, and thus the resulting channel from Tx to Rx in ID mode can be modeled by an erasure AWGN channel. In addition, the ST code is applied over $N$ RBs with TS instead of $N_t$ antennas with PS. This is because the use of $N$ RBs transforms the $N_t \times 1$ constant MISO channel ${{\bf{h}}}$ into an $N \times 1$ fading MISO channel specified by ${\bf{a}}\left( k \right)$'s in each transmission block. For all $K$ sub-blocks in TS, we consider the use of a ST code of length $L$ denoted by matrix ${{\bf{X}}^{({\rm{T}})}} \in {{\mathbb{C}}^{{L} \times N}}$. For convenience, we express ${{\bf{X}}^{({\rm{T}})}} = [ {{\bf{x}}_1^{({\rm{T}})}\,\,{\bf{x}}_2^{({\rm{T}})}\,\, \cdots \,\,{\bf{x}}_L^{({\rm{T}})}} ]^T$, where ${{\bf{x}}_l^{({\rm{T}})}} \in {{\mathbb{C}}^{{N} \times 1}}$, $1 \le l \le L$, denotes the $l$th transmitted signal vector in each ST coded block. The covariance matrix for ${\bf{x}}_l^{({\rm{T}})}$ is given by ${\bf{S}}_{{\bf{x}},l}^{(\rm{T})} = \mathbb E [ {{\bf{x}}_l^{({\rm{T}})}{{( {{\bf{x}}_l^{({\rm{T}})}} )}^H}} ] = \frac{P}{N}{{\bf{I}}_N}$, $\forall l$, to be consistent with ${{\bf{S}}_{\bf{x}}} = \frac{P}{{{N_t}}}{{\bf{I}}_{{N_t}}}$. Similar to ${\bf{X}}^{(\rm{P})}$ in the case of PS, ${{\bf{X}}^{(\rm{T})}}$ is assumed to be a capacity-achieving ST code for an equivalent MISO channel with $N$ transmitting antennas.

   The received signal at each sub-block is used for either energy harvesting or information decoding according to (\ref{Eq_Opt_1_Solution}). For the $k$th sub-block, the received signal can thus be expressed by modifying (\ref{Eq_ReceivedSignal_General}) as
   \begin{equation}\label{Eq_RCS_Received_Signal}
      {\begin{array}{l}
      {\bf{y}}\left( k \right) = {{\bf{X}}^{(\rm{T})}} {\Phi ^{T}} \left( k \right) {\tilde{\bf{h}}} + {\bf{z}}\left( k \right) \\
      \,\,\,\,\,\,\,\,\,\,\,\,\,\, = \sqrt{\theta}\,{{\bf{X}}^{(\rm{T})}} {\bf{a}}{\left( k \right)} + {\bf{z}}\left( k \right),\\
      \end{array}}
   \end{equation}
   where ${\bf{y}}\left( k \right) \in {{\mathbb C}^{L \times 1}}$ and ${\bf{z}}\left( k \right) \in {{\mathbb C}^{L \times 1}}$ denote the received signal and noise vectors, respectively, with ${\bf{z}}\left( k \right) \sim {\mathcal{CN}}\left( {{\bf{0}},\sigma^2 {{\bf{I}}_L}} \right)$. When $\rho \left( k \right) = 0$, the amount of harvested power (i.e., harvested energy normalized by sub-block duration $T/K$) at the $k$th sub-block is derived using ${\bf{S}}_{{\bf{x}},l}^{(\rm{T})}$ as
   \begin{equation}\label{Eq_Energy_TBS1}
      {{Q^{(\rm{T})}}\left( k \right) = \frac{1}{L}{\mathbb E}\left[ {{{\left\| \sqrt{\theta}\,{{\bf{X}}^{(\rm{T})}}{{\bf{a}}{{\left( k \right)}}} \right\|}^2}} \right] = \theta P{A\left( k \right)}.}
   \end{equation}
   Furthermore, by assuming a capacity-achieving ST code, the achievable rate with TS at the $k$th sub-block when $\rho \left( k \right) = 1$ can be expressed as
   \begin{equation}\label{Eq_Rate_TBS1}
      {{R^{(\rm{T})}}\left( k \right) = {\log _2}\left( {1 + \frac{\theta P A\left( k \right)}{\sigma^2}} \right).}
   \end{equation}

   \begin{figure}
      \centering
      \includegraphics[width=0.6\columnwidth]{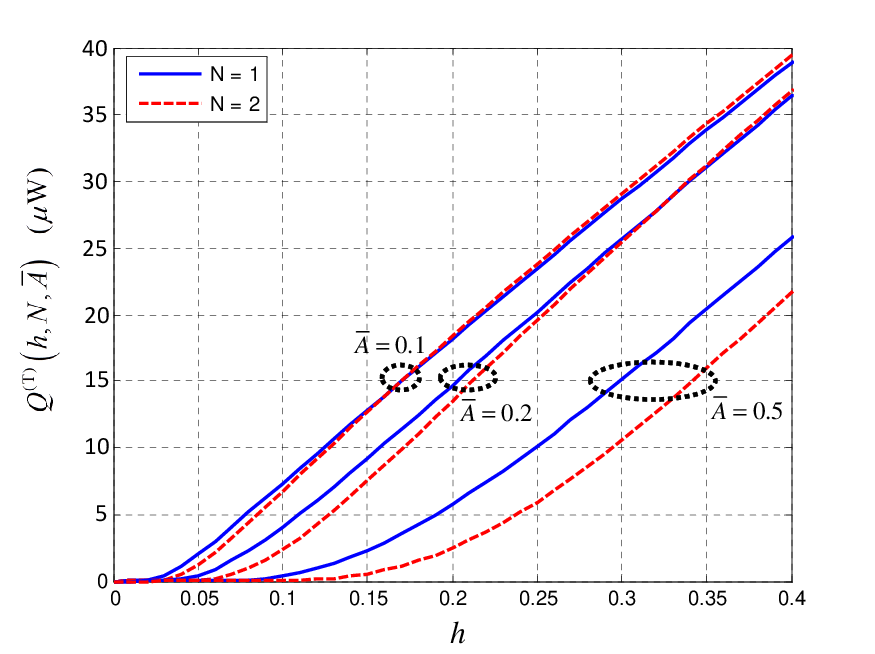}
      \caption{${Q^{(\rm{T})}}\left( {h,N,\bar A} \right)$ vs. $h$ with $P = 30$dBm, $N = 1, 2$, $\theta = 10^{-4}$, and $\bar A = 0.1$, $0.2$, $0.5$.}
      \label{Fig_HarvestedEnergy_vs_H}
   \end{figure}

   The amount of harvested energy in a transmission block is the sum of the energy harvested from all sub-blocks in the EH mode. Assuming that $K \to \infty$, the average harvested power in a transmission block for given $N$ RBs, threshold $\bar A$, and the realization of the normalized MISO channel $\bf{h}$ with $H = h$ can be obtained from (\ref{Eq_Energy_TBS1}) as
   \[
      {{Q^{(\rm{T})}}\left( {h,N,\bar A} \right) = \frac{1}{T}\mathop {\lim }\limits_{K \to \infty } \sum\limits_{k = 1}^K {\left( {1 - \rho \left( k \right)} \right)\frac{T \times {Q^{(\rm{T})}}\left( k \right)}{K}}}
   \]
   \begin{equation}\label{Eq_Energy_TBS_1_1}
      \,\,\,\,\,\, = \mathbb{E}\left[ {\left( {1 - \rho \left( k \right)} \right)\theta PA\left( k \right)} \right].\,
   \end{equation}
   In this section, Gaussian RBs\footnote{Alternative RB designs will be studied later in Section \ref{Sec:OtherBeams}.} are assumed to generate artificial channel fading, i.e., ${\boldsymbol{\phi} _n}\left( k \right) \sim {\mathcal{CN}} ( {{\bf{0}},\frac{1}{N_t}{{\bf{I}}_{{N_t}}}} )$. It can be easily verified that ${\bf{a}}\left( k \right) \sim {\mathcal{CN}}\left( {{\bf{0}},H{{\bf{I}}_N}} \right)$ for a given $H$, and $A \left( k \right)$ is thus a chi-square random variable with $2N$ degrees-of-freedom. With $N$ RBs and conditioned on a given normalized MISO channel realization ${\bf{h}}$ with $H = h$, the probability density function (PDF) of $A := A\left( k \right)$, $\forall k,$ and the cumulative distribution function (CDF) of $A$ are given, respectively, by \cite{Proakis}
   \begin{equation}\label{Eq_TBS_pdf}
      {{f_{A\left| H \right.}^{(N)}}\left( {a\left| h \right.} \right) = \frac{1}{{{{\left( {h/N} \right)}^N}\Gamma \left( N \right)}}{a^{N - 1}}{e^{ - \left( {N/h} \right)a}},}
   \end{equation}
   \begin{equation}\label{Eq_TBS_CDF}
      {{F_{A\left| H \right.}^{(N)}}\left( {a \left| h \right.} \right) = 1 - \frac{{\Gamma \left( {N,\frac{{Na}}{h}} \right)}}{{\Gamma \left( {N} \right)}} ,}
   \end{equation}
   where $\Gamma \left( x \right) = \int_0^\infty  {{t^{x - 1}}{e^{ - t}}dt}$ and $\Gamma \left( {\alpha ,x} \right) = \int_x^\infty  {{t^{\alpha  - 1}}{e^{ - t}}dt}$ represent the Gamma function and incomplete Gamma function, respectively. From (\ref{Eq_Energy_TBS_1_1}) and (\ref{Eq_TBS_pdf}), ${Q^{(\rm{T})}}\left( {h,N,\bar A} \right)$ with Gaussian RBs can thus be obtained as
   \begin{equation}\label{Eq_Energy_TBS2}
      {{Q^{(\rm{T})}}\left( {h,N,\bar A} \right) = \int_{\bar A}^\infty  {\theta Pa{f_{A\left| H \right.}^{(N)}}\left( {a\left| h \right.} \right)da}}
   \end{equation}
   \begin{equation}\label{Eq_Energy_TBS3}
      {\,\,\,\,\,\,\,\,\,\,\,\,\,\,\,\,\,\,\,\,\,\,\,\,\,\,\,\,\,\,\,\, = \theta Ph\frac{{\Gamma \left( {N + 1,\frac{{N\bar A}}{h}} \right)}}{{\Gamma \left( {N + 1} \right)}}, }
   \end{equation}
   where (\ref{Eq_Energy_TBS3}) can be obtained by applying (\ref{Eq_TBS_pdf}) and \cite[3.351-2]{TableOfIntegral} to (\ref{Eq_Energy_TBS2}). For an illustration, Fig. \ref{Fig_HarvestedEnergy_vs_H} shows ${Q^{(\rm{T})}}\left( {h,N,\bar A} \right)$ versus different values of $h$ when $N = 1$, $2$ and $\bar A = 0.1$, $0.2$, $0.5$, assuming 40dB signal power attenuation due to large-scale fading, i.e., $\theta = 10^{-4}$, with the carrier frequency and the distance between Tx and Rx given by $900$MHz and $5$ meters. The transmit power at Tx is set to be $P = 30$dBm. It is observed that ${Q^{(\rm{T})}}\left( {h,N,\bar A} \right)$ decreases with increasing $\bar A$ when $N$ and $h$ are both fixed, which is in accordance with (\ref{Eq_Energy_TBS3}). Moreover, when $N$ and $\bar A$ are both fixed, ${Q^{(\rm{T})}}\left( {h,N,\bar A} \right)$ is observed to increase monotonically with $h$. This is because ${F_{A\left| H \right.}^{(N)}}\left( {\bar A \left| h \right.} \right)$ in (\ref{Eq_TBS_CDF}) decreases with increasing $h$, and thus $1 - F_{A\left| H \right.}^{(N)}( {\bar A\left| h \right.} )$, which is the percentage of the received sub-blocks allocated to EH mode in each block, increases. Thus, the amount of harvested power in each block increases with $h$ thanks to the increased number of sub-blocks assigned to EH mode, as well as the increased average channel power $h$, as can be inferred from (\ref{Eq_Energy_TBS3}).

   Furthermore, when $h$ and $\bar A$ are both fixed, ${Q^{(\rm{T})}}\left( {h,N,\bar A} \right)$ is observed to decrease with increasing $N$ when $h$ is small, but increase with $N$ when $h$ is sufficiently large. This is because, as inferred from (\ref{Eq_TBS_pdf}) and (\ref{Eq_TBS_CDF}), the artificial channel fading is more substantial when smaller number of RBs, $N$, is used, although the same average channel power is given as $h$. Given $1 \le N \le N_t$, it can be shown that $F_{A\left| H \right.}^{(N)}( {\bar A\left| h \right.} )$ in (\ref{Eq_TBS_CDF}) increases with $N$ when $h$ is small, and thus larger power is harvested with smaller number of RBs. In contrast, it can also be shown that $F_{A\left| H \right.}^{(N)}( {\bar A\left| h \right.} )$ decreases with increasing $N$ when $h$ is larger than a certain threshold, and thus more power is harvested with larger number of RBs. Similarly, we can verify that ${Q^{(\rm{T})}}\left( {h,N,\bar A} \right)$ increases with $N$ when $\bar A$ is small, but decreases with increasing $N$ when $\bar A$ is sufficiently large.

   Next, the achievable rate in a block for given $N$, $\bar A$, and $h$ can be derived from (\ref{Eq_Rate_TBS1}) and (\ref{Eq_TBS_pdf}) as
   \[
      {R^{(\rm{T})}}\left( {h,N,\bar A} \right) = {\mathbb E}\left[ \rho \left( k \right) {{\log }_2}\left( 1 + \frac{\theta P A \left( k \right)}{\sigma^2}  \right) \right]
   \]
   \begin{equation}\label{Eq_Rate_TBS2}
      {\,\,\,\,\,\,\,\,\,\,\,\,\,\,\,\,\,\,\,\,\,\,\,\,\,\,\,\,\,\,\,\,\,\,\,\,\,\,\,\,\,\,\,\,\,\,\,\,\,\,\,\,\,\, = \int_0^{\bar A} {{{\log }_2}\left( 1 + \frac{\theta P}{\sigma^2} a \right){f_{A\left| H \right.}^{(N)}}\left( {a\left| h \right.} \right)da.}}
   \end{equation}
   With ${f_{A\left| H \right.}^{(N)}}\left( {a\left| h \right.} \right)$ given in (\ref{Eq_TBS_pdf}), it is in general difficult to obtain a unified closed-form expression of (\ref{Eq_Rate_TBS2}) for arbitrary values of $N$. However, it is possible to derive closed-form expressions for (\ref{Eq_Rate_TBS2}) for some special values of $N$. For example, ${R^{(\rm{T})}}\left( {h,1,\bar A} \right)$ and ${R^{(\rm{T})}}\left( {h,2,\bar A} \right)$ for $N = 1$ and $2$, respectively, can be derived in closed-form in Appendix \ref{App_Rate_N1_Derivation}. Fig. \ref{Fig_AchievableRate_vs_H} shows ${R^{({\rm{T}})}}\left( {h,N,\bar A} \right)$ versus different values of $h$ when $N = 1$, $2$ and $\bar A = 0.1$, $0.2$, $0.5$ with the same setup as for Fig. \ref{Fig_HarvestedEnergy_vs_H} with $\theta = 10^{-4}$ and $P = 30$dBm. It is further assumed that the bandwidth of the transmitted signal is $10$MHz, and receiver noise is white Gaussian with power spectral density $-110$dBm/Hz or $-40$dBm over the entire bandwidth of $10$MHz. It is observed that ${R^{(\rm{T})}}\left( {h,N,\bar A} \right)$ increases with $\bar A$ when $N$ and $h$ are both fixed, which is in accordance with (\ref{Eq_Rate_TBS2}). Moreover, by the opposite argument of the explanation for Fig. \ref{Fig_HarvestedEnergy_vs_H}, when $h$ and $\bar A$ are both fixed, ${R^{(\rm{T})}}\left( {h,N,\bar A} \right)$ is observed to increase with $N$ when $h$ is small or $\bar A$ is large, but decrease with increasing $N$ when $h$ or $\bar A$ is sufficiently large/small.

   \begin{figure}
      \centering
      \includegraphics[width=0.6\columnwidth]{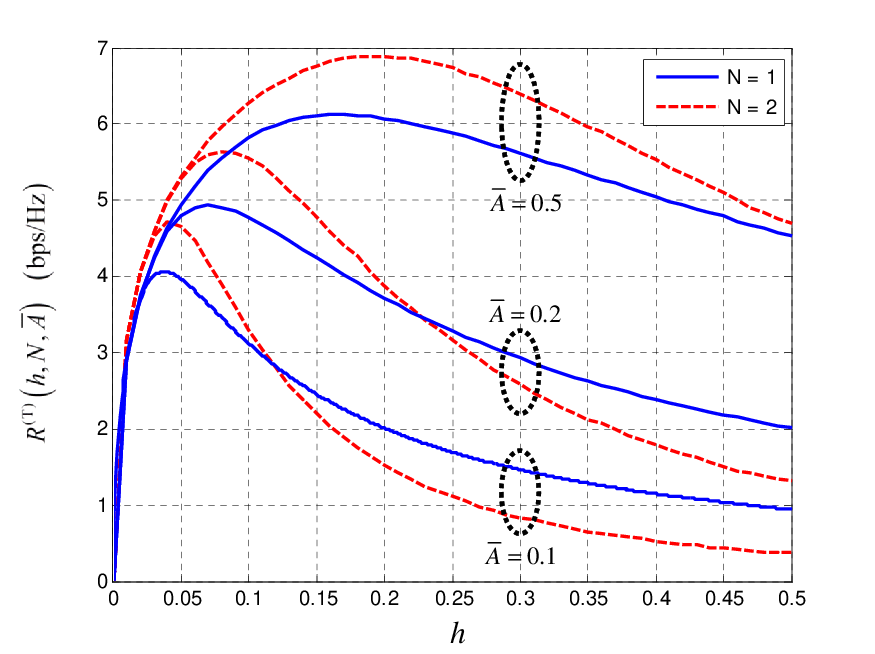}
      \caption{${R^{(\rm{T})}}\left( {h, N, \bar A} \right)$ vs. $h$ with $P = 30$dBm, $N = 1, 2$, $\theta = 10^{-4}$, and $\bar A = 0.1$, $0.2$, $0.5$.}
      \label{Fig_AchievableRate_vs_H}
   \end{figure}

   However, different from ${Q^{(\rm{T})}}\left( {h,N,\bar A} \right)$ in Fig. \ref{Fig_HarvestedEnergy_vs_H} which is a monotonically increasing function of $h$, it is observed in Fig. \ref{Fig_AchievableRate_vs_H} that ${R^{(\rm{T})}}\left( {h,N,\bar A} \right)$ in general first increases with $h$, and then decreases with increasing $h$ for given $N$ and $\bar A$. The reason is as follows. When $h \to 0$, from (\ref{Eq_Rate_TBS2}), we have ${R^{(\rm{T})}}\left( {h,N,\bar A} \right) \to {{{\log }_2}\left( {1 + \frac{\theta P}{\sigma^2} h} \right)}$; thus, ${R^{(\rm{T})}}\left( {h,N,\bar A} \right)$ increases with $h$. However, when $h \to \infty$, ${f_{A\left| H \right.}^{(N)}}\left( {\bar A\left| h \right.} \right) \to 0$ for any finite $0 \le a \le \bar A$, and thus ${R^{(\rm{T})}}\left( {h,N,\bar A} \right) \to 0$; therefore, ${R^{(\rm{T})}}\left( {h,N,\bar A} \right)$ should decrease with increasing $h$ when $h$ is sufficiently large.

   \subsection{Rate-Energy Performance Comparison}\label{NumericalExample1}
   As in \cite{Liu} and \cite{Zhang}, there exist rate-energy (R-E) trade-offs in both PS and TS schemes for information and energy transfer. R-E trade-offs in PS and TS can be characterized by setting different values of $\tau$ and $\bar A$, respectively. Fig. \ref{Fig_R_E_Tradeoff_Nt_2} shows R-E trade-offs in PS and TS for $N_t = 2$ and a constant MISO channel ${\bf{h}} = {\left[ {1.0\,\,\,\,0.56} \right]^T}$, with the same channel setup as for Figs. \ref{Fig_HarvestedEnergy_vs_H} and \ref{Fig_AchievableRate_vs_H}. For PS, ${\bf{X}}^{(\rm{P})}$ is generated by Alamouti code with $L = 2$ \cite{Alamouti}. For TS, a scalar code cascaded by one single RB is applied when $N = 1$, while the Alamouti code with two RBs is applied when $N=2$. The harvested power is denoted by $Q$. It is observed that TS yields the best R-E trade-off with $N=1$ when $Q_{th}^{\left( 1 \right)} \le Q \le h$, and with $N=2$ when $Q_{th}^{\left( 2 \right)} \le Q < Q_{th}^{\left( 1 \right)}$, while PS yields the best R-E trade-off when $0 \le Q < {Q_{th}^{(2)}}$, where $Q_{th}^{\left( 1 \right)}$ and $Q_{th}^{\left( 2 \right)}$ are shown in Fig. \ref{Fig_R_E_Tradeoff_Nt_2}. Note that at $Q = 0$, i.e., no EH is required as in the conventional MISO system with WIT only, PS achieves higher rate than TS since artificial channel fading by random beamforming degrades the AWGN channel capacity. However, when the harvested power exceeds certain thresholds, i.e., $Q_{th}^{(2)}$ and $Q_{th}^{(1)}$, TS with $N = 2$ RBs and $N = 1$ RB achieves the best rate performance for a given power harvesting target, respectively. This demonstrates the unique usefulness of random beamforming in a multi-antenna SWIPT system even with constant AWGN channels.

   \begin{figure}
      \centering
      \includegraphics[width=0.6\columnwidth]{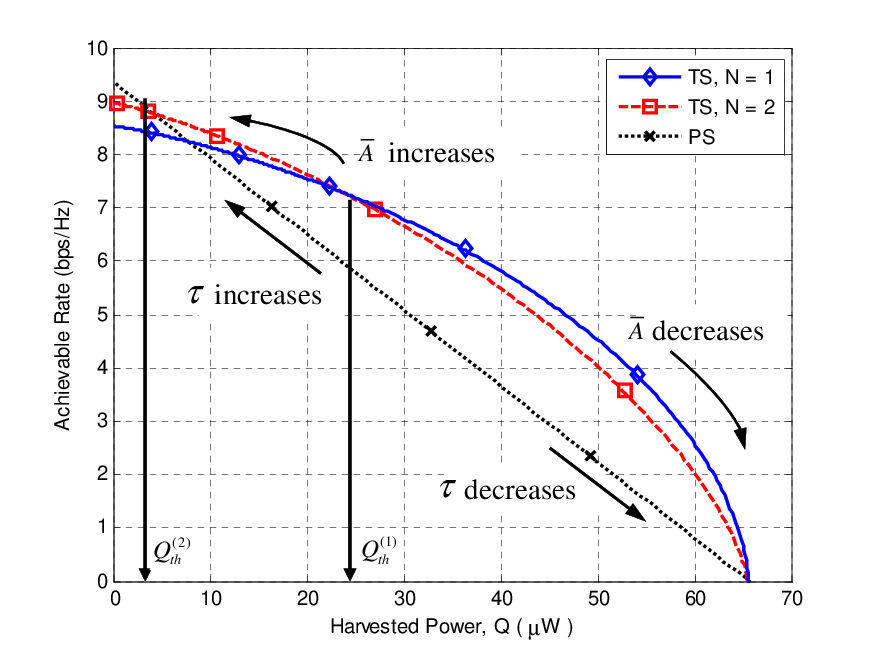}
      \caption{Trade-offs between achievable rate and harvested power when $P=30$dBm, $N_t = 2$, $\theta = 10^{-4}$, and ${\bf{h}} = {\left[ {1.0\,\,\,\,0.56} \right]^T}$.}
      \label{Fig_R_E_Tradeoff_Nt_2}
   \end{figure}

   It is worth noting that for TS larger information rate is achieved with $N = 1$ when $Q_{th}^{\left( 1 \right)} \le Q \le h$, but with $N=2$ otherwise. This can be explained as follows. For a given $h$, it can be shown from (\ref{Eq_Energy_TBS3}) that $\bar A \to 0$ when $Q \to h$. Thus, with sufficiently small $\bar A$, we have ${Q^{(\rm{T})}}\left( {h,1,\bar A} \right) \approx {Q^{(\rm{T})}}\left( {h,2,\bar A} \right)$ (note that ${Q^{(\rm{T})}}\left( {h,2,\bar A} \right)$ is sightly larger than ${Q^{(\rm{T})}}\left( {h,1,\bar A} \right)$ for small $\bar A$ as discussed for Fig. \ref{Fig_HarvestedEnergy_vs_H}; but the gap between them is negligible as shown in Fig. \ref{Fig_HarvestedEnergy_vs_H} with $\bar A = 0.1$). On the other hand, with small $\bar A$, it can be shown from (\ref{Eq_TBS_pdf}) that ${f_{A\left| H \right.}^{(1)}}\left( {a\left| h \right.} \right) > {f_{A\left| H \right.}^{(2)}}\left( {a\left| h \right.} \right)$, $0 \le a \le \bar A$, and thus ${R^{(\rm{T})}}\left( {h,1,\bar A} \right) > {R^{(\rm{T})}}\left( {h,2,\bar A} \right)$ from (\ref{Eq_Rate_TBS2}), as discussed for Fig. \ref{Fig_AchievableRate_vs_H}. Therefore, TS with $N = 1$ achieves larger information rate than $N = 2$ when $Q$ is sufficiently large. In contrast, as $Q \to 0$, we have $\bar A \to \infty$ from (\ref{Eq_Energy_TBS3}). Then, it can be shown that ${R^{(\rm{T})}}\left( {h,1,\infty} \right) < {R^{(\rm{T})}}\left( {h,2,\infty} \right)$ since the ergodic capacity of a fading MISO channel increases with the number of transmit antennas. Therefore, for TS larger information rate is achieved with $N = 2$ than $N = 1$ when $Q$ is smaller than a certain threshold.

\section{Performance Analysis in Fading MISO Channel}\label{Sec:PerformanceAnalysis}
In this section, the R-E performances of TS and PS schemes are further analyzed in fading MISO channels. It is assumed that the small-scale MISO channel from Tx to each Rx follows independent and identically distributed (i.i.d.) Rayleigh fading with ${\bf{h}} \sim \mathcal{CN}\left( {{\bf{0}},{{\bf{I}}_{{N_t}}}} \right)$, and thus $H = \frac{1}{{{N_t}}}\left\| {\bf{h}} \right\|^2$ is a chi-square random variable with $2N_t$ degrees-of-freedom, with the following PDF and CDF \cite{Proakis}:
\begin{equation}\label{Eq_pdf_H}
   {{f_H}\left( h \right) = \frac{{{N_t}^{{N_t}}}}{{\Gamma \left( {{N_t}} \right)}}{h^{{N_t} - 1}}{e^{ - {N_t}h}},}
\end{equation}
\begin{equation}\label{Eq_CDF_H}
   {{F_H}\left( h \right) = 1 - \frac{{\Gamma \left( {{N_t},{N_t}h} \right)}}{{\Gamma \left( {{N_t}} \right)}}.}
\end{equation}

In practice, it is possible for Rxs to change $\bar A$ for TS or $\tau$ for PS with the fading MISO channel $\bf{h}$ for different transmission blocks; however, this incurs additional complexity at Rx. For simplicity, it is assumed in this paper that $\bar A$ and $\tau$ are set to be fixed values for all Rxs over different realizations of $\bf{h}$ for a given $\theta$.

   \subsection{Achievable Average Information Rate}
   We consider that the performance of information transfer is measured by the achievable average rate over fading channels. Given $N$ and $\bar A$, the achievable average rate of TS is denoted by ${\bar R^{\left( {\rm{T}} \right)}}\left( {N,\bar A} \right) = {{\mathbb E_H}\left[ {{R^{\left( {\rm{T}} \right)}}\left( {h,N,\bar A} \right)} \right]}$, where ${{R^{\left( {\rm{T}} \right)}}\left( {h,N,\bar A} \right)}$ is given by (\ref{Eq_Rate_TBS2}) for a given $h$. However, it is difficult to obtain the closed-form expressions for ${\bar R^{\left( {\rm{T}} \right)}}\left( {N,\bar A} \right)$'s using (\ref{Eq_Rate_TBS2}) and (\ref{Eq_pdf_H}) for any given $N$, $1 \le N \le N_t$.

   \begin{figure}
      \centering
      \includegraphics[width=0.6\columnwidth]{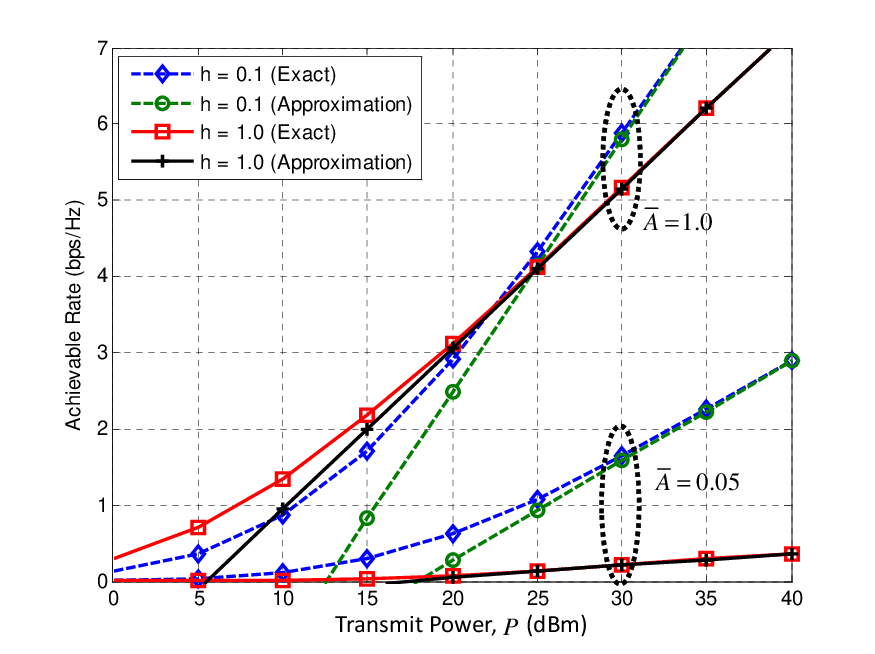}
      \caption{Plot of ${R^{(\rm{T})}}\left( {h,N,\bar A} \right)$ with $N = 1$, $h = 0.1$ and $1.0$, $\bar A = 0.05$ and $1.0$.}
      \label{Fig_AchievableRate_Exact_Approx}
   \end{figure}

   Note that in practice, SWIPT systems usually operate with large transmit power $P$ due to the requirement of energy transfer, resulting in large $\frac{\theta P}{\sigma^2}$, (e.g., $\frac{\theta P}{\sigma^2} = 30$dB with the setup for Fig. \ref{Fig_AchievableRate_vs_H}). It is also worth noting that as $P \to \infty$, ${\log _2}\left( {1 + \frac{\theta P}{\sigma^2} a} \right) = {\log _2}\left( {\frac{\theta P a}{\sigma^2}} \right) + o\left( {{{\log }_2} P } \right)$ for given $a > 0$,\footnote{$f\left( x \right) = o\left( {g\left( x \right)} \right)$ as $x \to x_0$ represents that $\mathop {\lim }\limits_{x \to {x_0}} \frac{{f\left( x \right)}}{{g\left( x \right)}} = 0$, meaning intuitively that $f\left( x \right) \ll g\left( x \right)$ as $x \to x_0$.} resulting in $\mathop {\lim }\limits_{P  \to \infty } {\log _2}(1 + \frac{\theta P a}{\sigma^2}) = {\log _2}(\frac{\theta P a}{\sigma^2})$. Therefore, $\mathop {\lim }\limits_{P \to \infty } {R^{(T)}}\left( {h,N,\bar A} \right) = \mathop {\lim }\limits_{P  \to \infty } \int_0^{\bar A} {{{\log }_2}\left( {\frac{\theta P a}{\sigma^2}} \right)f_{A\left| H \right.}^{(N)}\left( {a\left| h \right.} \right)da}$, and as $P$ is sufficiently large, ${R^{(\rm{T})}}\left( {h,N,\bar A} \right)$ in (\ref{Eq_Rate_TBS2}) with $\bar A > 0$ can be approximated as
   \begin{equation}\label{Eq_Rate_TBS3}
      {{R^{(\rm{T})}}\left( {h,N,\bar A} \right) \approx {F_{A\left| H \right.}^{(N)}}\left( {\bar A \left| h \right.} \right){\log _2}\left( \frac{\theta P}{\sigma^2}  \right) + C_0 \left( {h,N,\bar A} \right),}
   \end{equation}
   where ${F_{A\left| H \right.}}\left( {\bar A \left| h \right.} \right) = \int_0^{\bar A} {{f_{A\left| H \right.}}\left( {a \left| h \right.} \right)da}$ and $C_0 \left( {h,N,\bar A} \right) = \int_0^{\bar A} {{{\log }_2}\left( a \right){f_{A\left| H \right.}}\left( {a\left| h \right.} \right)da}$, which is a constant not related to $P$. Please refer to Appendix \ref{App_C0_Derivation} for detailed derivation of $C_0 \left( {h,N,\bar A} \right)$. Note that the right-hand side of (\ref{Eq_Rate_TBS3}) is a lower bound on ${R^{(\rm{T})}}\left( {h,N,\bar A} \right)$, but approximates ${R^{(\rm{T})}}\left( {h,N,\bar A} \right)$ tightly with sufficiently large $P$. Fig. \ref{Fig_AchievableRate_Exact_Approx} shows ${R^{(\rm{T})}}\left( {h,N,\bar A} \right)$ and its approximation by (\ref{Eq_Rate_TBS3}) versus $P$ for different values of $h$ and $\bar A$ with the same setup as for Fig. \ref{Fig_AchievableRate_vs_H} and $N = 1$. It is observed that the approximation in (\ref{Eq_Rate_TBS3}) is more accurate as $h$ and/or $\bar A$ increases. It is also observed that the gap between the achievable rate and its approximation becomes negligible when $P \ge 30$dBm even with moderate values of $h = 0.1$ and $\bar A = 0.05$.

   With the approximation of ${R^{(\rm{T})}}\left( {h,N,\bar A} \right)$ by (\ref{Eq_Rate_TBS3}), we can characterize the asymptotic behavior of ${\bar R^{\left( {\rm{T}} \right)}}\left( {N,\bar A} \right)$ as $P$ becomes large by investigating its pre-log scaling factor, which is given by the following proposition.

   \begin{proposition}\label{Proposition_Avg_Scaling}
      Given $1 \le N \le N_t$ and $\bar A \ge 0$, the achievable average rate for TS over the i.i.d. Rayleigh fading MISO channel is obtained as ${\bar R^{\left( {\rm{T}} \right)}}\left( {N,\bar A} \right) = {\Delta^{\left( {\rm{T}} \right)}}\left( {N,\bar A} \right){\log _2}\left( P \right) + o\left( {\log _2} \, P \right)$ with $P \rightarrow \infty$, where
      \begin{equation}\label{Eq_TBS_AvgRate_Scaling}
         {{\Delta^{\left( {\rm{T}} \right)}}\left( {N,\bar A} \right) \buildrel \Delta \over = \mathop {\lim }\limits_{P \to \infty } \frac{{\bar R^{\left( {\rm{T}} \right)}}\left( {N,\bar A} \right)}{{{{\log }_2}\, P}} = {F_A^{(N)}}\left( \bar A \right),}
      \end{equation}
      with ${F_A^{(N)}}\left( a \right) = \mathbb E_H \left[{F_{A\left| H \right.}^{(N)}}\left( {a\left| h \right.} \right)\right]$ denoting the unconditional CDF of $A$ after averaging over the fading distribution, which can be further expressed as
      \begin{equation}\label{Eq_Marginal_CDF}
         {{F_A^{(N)}}\left( a \right) = 1 - \frac{2}{{\Gamma \left( {{N_t}} \right)}}\sum\limits_{k = 0}^{N - 1} {\frac{{{ \left( \beta \left( a \right) \right)^{{N_t} + k}}}}{{k!}}{K_{{N_t} - k}}\left( {2\beta \left( a \right) } \right),}}
      \end{equation}
      where $\beta \left( a \right)  \buildrel \Delta \over =  \sqrt{{N_t}Na}$, and ${K_\delta }\left( x \right)$ denotes the second-kind modified Bessel function
      \[ {{K_\delta }\left( x \right) = \frac{\pi }{2}\frac{{{I_{ - \delta }}\left( x \right) - {I_\delta }\left( x \right)}}{{\sin \left( {\delta x} \right)}},} \]
      with ${I_\delta }\left( x \right)$ denoting the first-kind modified Bessel function
      \[ {{I_\delta }\left( x \right) = \sum\limits_{m = 0}^\infty  {\frac{1}{{m!\Gamma \left( {m + \delta  + 1} \right)}}{{\left( {\frac{x}{2}} \right)}^{2m + \delta }}} .} \]
   \end{proposition}
   \begin{proof}
      Please refer to Appendix \ref{App_Proof_Proposition_Avg_Scaling}.
   \end{proof}

   \begin{remark}\label{Remark_RateScaling_vs_CDF}
      In the fading MISO channel, ${F_A^{(N)}}\left( \bar A \right)$ denotes the percentage of sub-blocks allocated to ID mode for TS. From Proposition \ref{Proposition_Avg_Scaling}, it is inferred that ${F_A^{(N)}}\left( \bar A \right)$ is also the pre-log rate scaling factor of the asymptotic achievable average information rate over the MISO fading channel for TS with given $\bar A$ and $N$.
   \end{remark}

   Fig. \ref{Fig_Marginal_CDF} shows ${F_A^{(N)}}\left( \bar A \right)$ versus $\bar A$ for TS with $N_t = 4$ when $\bf{h}$ follows i.i.d. Rayleigh fading. From Fig. \ref{Fig_Marginal_CDF}, it is observed that the rate scaling factor ${F_A^{(N)}}\left( \bar A \right)$ for TS increases with decreasing $N$ when $\bar A$ is small, but decreases with $N$ when $\bar A$ is sufficiently large. As a result, ${\bar R^{\left( {\rm{T}} \right)}}\left( {N,\bar A} \right)$ scales faster with increasing $P$ for smaller value of $N$ when $\bar A$ is small, but scales slower with $P$ when $\bar A$ becomes large.

   \begin{figure}
      \centering
      \includegraphics[width=0.65\columnwidth]{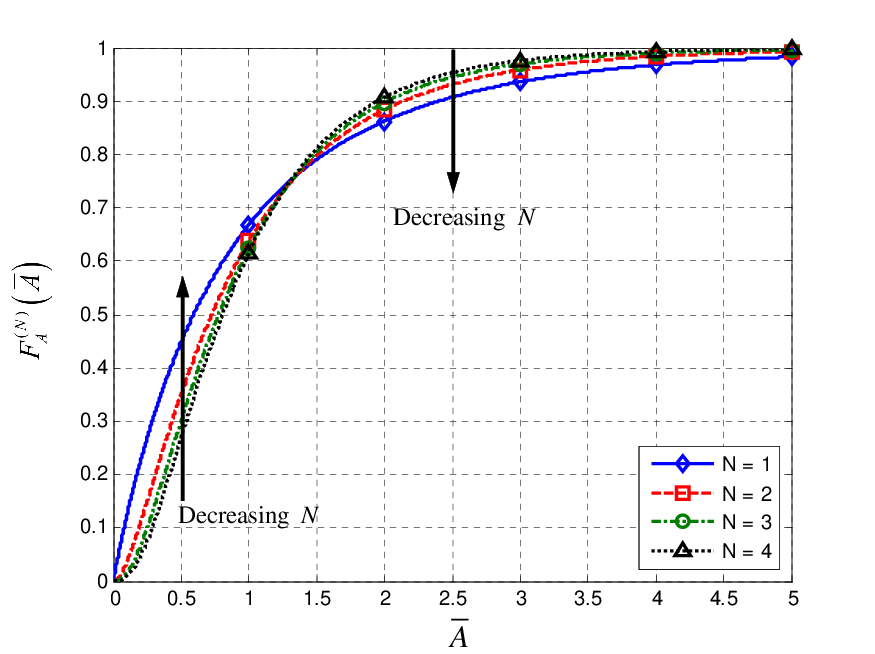}
      \caption{$F_A^{(N)}\left( \bar A \right)$ vs. $\bar A$ when $N_t = 4$.}
      \label{Fig_Marginal_CDF}
   \end{figure}

   On the other hand, the rate scaling factor for PS in the i.i.d. Rayleigh fading MISO channel can be determined from (\ref{Eq_Rate_TimeSharing}) and (\ref{Eq_pdf_H}) as
   \begin{equation}\label{Eq_PreLog_OOS}
      {{\Delta^{({\rm{P}})}}\left( \tau  \right) = \mathop {\lim }\limits_{P \to \infty } \frac{{\mathbb E_H\left[ {{R^{({\rm{P}})}}\left( {h,\tau } \right)} \right]}}{{{{\log }_2} P }} = \tau .}
   \end{equation}

   \subsection{Average Harvested Power}
   In this subsection, we study the average harvested power over the i.i.d. Rayleigh fading MISO channel by TS, defined as $\bar Q^{(\rm{T})}\left( {N,\bar A} \right)  = \mathbb E_{H}\left[ {{Q^{\left( {\rm{T}} \right)}}\left( {h,N,\bar A} \right)} \right]$, where ${{Q^{\left( {\rm{T}} \right)}}\left( {h,N,\bar A} \right)}$ is given by (\ref{Eq_Energy_TBS3}).

   \begin{proposition}\label{Proposition_AvgEnergy}
      In the i.i.d. Rayleigh fading MISO channel, for given $\bar A$ and $N$, the average harvested power for TS is given by
      \begin{equation}\label{Eq_AvgEnergy_Slope}
         {\bar Q^{(\rm{T})}\left( {N,\bar A} \right)  =  \theta P\frac{{2}}{{\Gamma \left( {{N_t}} \right)}}\sum\limits_{k = 0}^N {\frac{{\left( \beta \left( \bar A\right) \right)^{{{N_t} + k}}}}{{k!}}\sqrt{{ {\frac{{N\bar A}}{{{N_t}}}} }}{K_{{N_t} - k + 1}}\left( {2 \beta \left( \bar A \right)  } \right),}}
      \end{equation}
      where $\beta \left( a \right)$ and ${K_\delta }\left( x \right)$ are defined in Proposition  \ref{Proposition_Avg_Scaling}.
   \end{proposition}
   \begin{proof}
      Please refer to Appendix \ref{Proof_Proposition_AvgEnergy}.
   \end{proof}

   For convenience, we term ${\Pi ^{({\rm{T}})}}\left( {N,\bar A} \right) = {Q^{({\rm{T}})}}\left( {N,\bar A} \right)/({\theta P})$ as the \emph{power scaling factor} for TS with increasing $P$. Notice that $0 \le {\Pi ^{({\rm{T}})}}\left( {N,\bar A} \right) \le 1$. Fig. \ref{Fig_AvgEnergy_vs_A} shows ${\Pi ^{({\rm{T}})}}\left( {N,\bar A} \right)$ versus different values of $\bar A$ with $N_t = 4$. It is observed that the power scaling factor ${\Pi ^{({\rm{T}})}}\left( {N,\bar A} \right)$ for TS behaves in the opposite way of the rate scaling factor ${F_A^{(N)}}\left( \bar A \right)$, as compared to Fig. \ref{Fig_Marginal_CDF}, i.e., ${\Pi ^{({\rm{T}})}}\left( {N,\bar A} \right)$ decreases with $N$ when $\bar A$ is small, but increases with decreasing $N$ when $\bar A$ is sufficiently large. As a result, for given $\theta$ and $P$, ${\bar Q^{\left( {\rm{T}} \right)}}\left( {N,\bar A} \right)$ behaves the same as ${\Pi ^{({\rm{T}})}}\left( {N,\bar A} \right)$.

   \begin{figure}
      \centering
      \includegraphics[width=0.6\columnwidth]{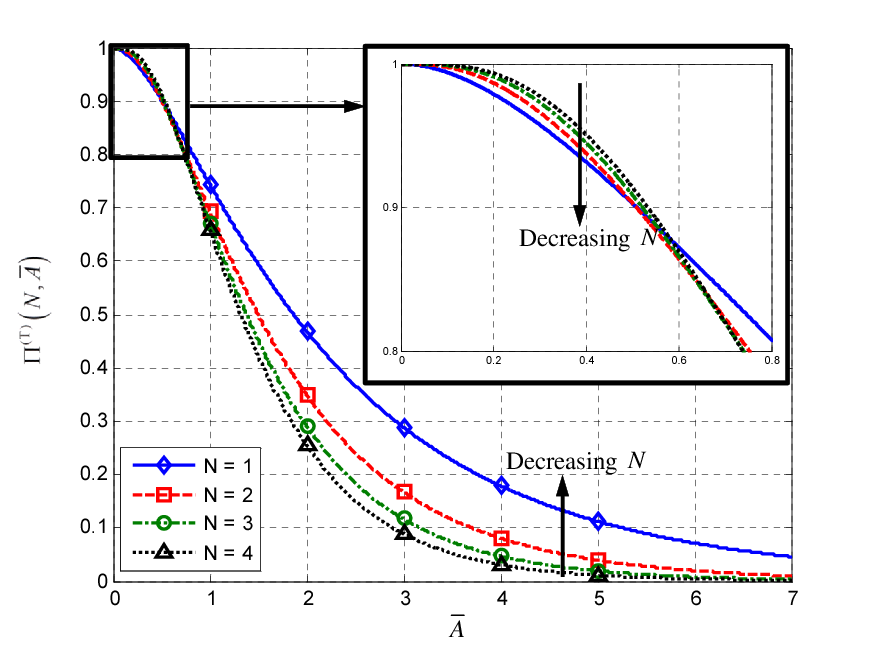}
      \caption{$\Pi^{(\rm{T})}\left( {N,\bar A} \right)$ vs. $\bar A$ when $N_t = 4$.}
      \label{Fig_AvgEnergy_vs_A}
   \end{figure}

   On the other hand, the power scaling factor for PS in the i.i.d. Rayleigh fading MISO channel can be easily obtained from (\ref{Eq_Energy_TimeSharing}) and (\ref{Eq_pdf_H}) as
   \begin{equation}\label{AvgEnergy_OOS}
      {\Pi^{(\rm{P})}\left( {\tau} \right) = {\mathbb{E}_H\left[ {{Q}^{(\rm{P})}\left( {h,\tau} \right)} \right] \mathord{\left/  {\vphantom {1 P}} \right. \kern-\nulldelimiterspace} \left( \theta P \right)}  = \left( {1 - \tau } \right), \,\,\, 0 \le \tau \le 1. }
   \end{equation}

   The rate and power scaling factors characterize the asymptotic rate-energy trade-off as $P \to \infty$. Given $1 \le N \le N_t$, for TS it is easily shown from (\ref{Eq_Marginal_CDF}) and (\ref{Eq_AvgEnergy_Slope}) that the rate scaling factor ${\Delta^{\left( {\rm{T}} \right)}}\left( {N, 0} \right) = 0$ and the power scaling factor ${\Pi^{({\rm{T}})}}( {N, 0} ) = 1$ at $\bar A = 0$, while ${\Delta^{\left( {\rm{T}} \right)}}\left( {N,\infty} \right) = 1$ and ${\Pi^{({\rm{T}})}}( {N,\infty} ) = 0$ at $\bar A \to \infty$. Note that for TS the distribution of the received channel power $A \left( k \right)$ at each sub-block becomes different according to $N$, and as a result different asymptotic rate-energy trade-off is achieved when $0 < {\Delta^{\left( {\rm{T}} \right)}}\left( {N,\bar A} \right) < 1$ and $0 < {\Pi^{\left( {\rm{T}} \right)}}\left( {N,\bar A} \right) < 1$. To characterize this trade-off, we have the following theorem.

   \begin{theorem}\label{Theorem_Optimality_N}
      In the i.i.d. Rayleigh fading MISO channel, given $1 \le N \le N_t$ and $0 < \bar A < \infty$ for TS scheme and $0 < \tau < 1$ for PS scheme, ${\Delta^{\left( {\rm{T}} \right)}}\left( {N,\bar A} \right) > {\Delta^{({\rm{P}})}}\left( \tau  \right)$ for a given power scaling factor $0 < {\Pi^{\left( {\rm{T}} \right)}}\left( {N,\bar A} \right) = {\Pi^{({\rm{P}})}}\left( \tau  \right) < 1$; furthermore, given $1 \le N < M \le N_t$ and $0 < \bar A_N, \bar A_M < \infty$ for TS schemes, ${\Delta^{\left( {\rm{T}} \right)}}\left( {N,\bar A_N} \right) > {\Delta^{\left( {\rm{T}} \right)}}\left( {M,\bar A_M} \right)$ for a given power scaling factor $0 < {\Pi^{\left( {\rm{T}} \right)}}\left( {N,\bar A_N} \right) = {\Pi^{\left( {\rm{T}} \right)}}\left( {M,\bar A_M} \right) < 1$.
   \end{theorem}
   \begin{proof}
      Please refer to Appendix \ref{App_Proof_Theorem_Optimality_N}.
   \end{proof}

   \begin{figure}
      \centering
      \includegraphics[width=0.6\columnwidth]{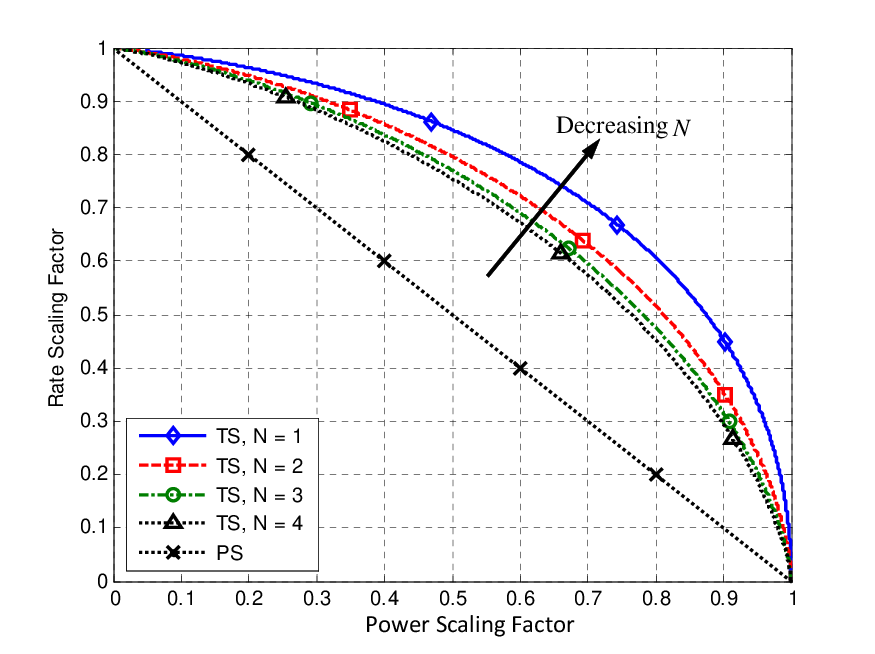}
      \caption{Rate vs. power scaling factors with $N_t = 4$.}
      \label{Fig_AvgEnergy_vs_PreLogFactor}
   \end{figure}

   Fig. \ref{Fig_AvgEnergy_vs_PreLogFactor} shows the rate scaling factor (${\Delta^{\left( {\rm{T}} \right)}}\left( {N,\bar A} \right)$ for TS and ${\Delta^{({\rm{P}})}}\left( \tau  \right)$ for PS) versus power scaling factor ($\Pi^{(\rm{T})}\left( {N,\bar A} \right)$ for TS and $\Pi^{(\rm{P})}\left( {\tau} \right)$ for PS) with $N_t = 4$. For a given $0 < \Pi^{(\rm{T})}\left( {N,\bar A} \right)  = {\Pi^{({\rm{P}})}}\left( \tau  \right) < 1$, the rate scaling factor of TS with $N = 1$, i.e., one single random beam, is the largest among all values of $N$. In addition, ${\Delta^{\left( {\rm{T}} \right)}}\left( {N,\bar A} \right)$ for TS decreases with increasing $N$, but is always larger than ${\Delta^{({\rm{P}})}}\left( \tau  \right)$ for PS. The above observations are in accordance with Theorem \ref{Theorem_Optimality_N}.

   \subsection{Power Outage Probability}
   In this subsection, we study the power outage probability with a given harvested power target $\hat Q$ at Rx, which is defined as $p_{Q,\,out} \mathop  = \limits^\Delta {\Pr \left( {Q < \hat Q} \right)}$ with $Q$ denoting the harvested power in one block. In particular, we are interested in characterizing the asymptotic behavior of $p_{Q,\,out}$ as $P \to \infty$, namely \emph{power diversity order}, which is defined as
   \begin{equation}\label{Eq_EnergyDiversity}
      {{d_Q} \buildrel \Delta \over = - \mathop {\lim }\limits_{P \to \infty } \frac{{\log {p_{Q,\,out}}}}{{{{\log }} P}}.}
   \end{equation}

   \begin{proposition}\label{Proposition_EnergyDiversity_TBS}
      In the i.i.d. Rayleigh fading MISO channel, for TS the power outage probability $p_{Q,out}^{({\rm{T}})}$ with $P \to \infty$ is approximated by
      \begin{equation}\label{Eq_EnergyDiversity_TBS}
         {p_{Q,out}^{({\rm{T}})} = \left\{ {\begin{array}{*{20}{c}}
         {{{\left( {{{\hat Q}}/ \left({\theta P}\right)} \right)}^{{N_t}}}}  \\
         {{{\left( {N\bar A{{\left( {\ln \left(\theta P\right)} \right)}^{ - 1}}} \right)}^{{N_t}}}}  \\
         \end{array}\begin{array}{*{20}{c}}
         {,\,\,\,\bar A = 0}  \\
         {,\,\,\,\bar A > 0.}  \\
         \end{array}} \right.}
      \end{equation}
   \end{proposition}
   \begin{proof}
      Please refer to Appendix \ref{App_Proof_Proposition_EnergyDiversity_TBS}.
   \end{proof}

   From (\ref{Eq_EnergyDiversity}) and (\ref{Eq_EnergyDiversity_TBS}), it can be verified that the power diversity order of TS is $d_Q^{({\rm{T}})} = N_t$ when $\bar A = 0$, i.e., no WIT is required, while $d_Q^{({\rm{T}})} = 0$ with a fixed $\bar A > 0$ when both WIT and WET are implemented, which means that although $p_{Q,out}^{({\rm{T}})}$ decreases with increasing $P$, the decrease of $p_{Q,out}^{({\rm{T}})}$ is much slower than increase of $P$ as $P \to \infty$.

   \begin{figure}
      \centering
      \includegraphics[width=0.58\columnwidth]{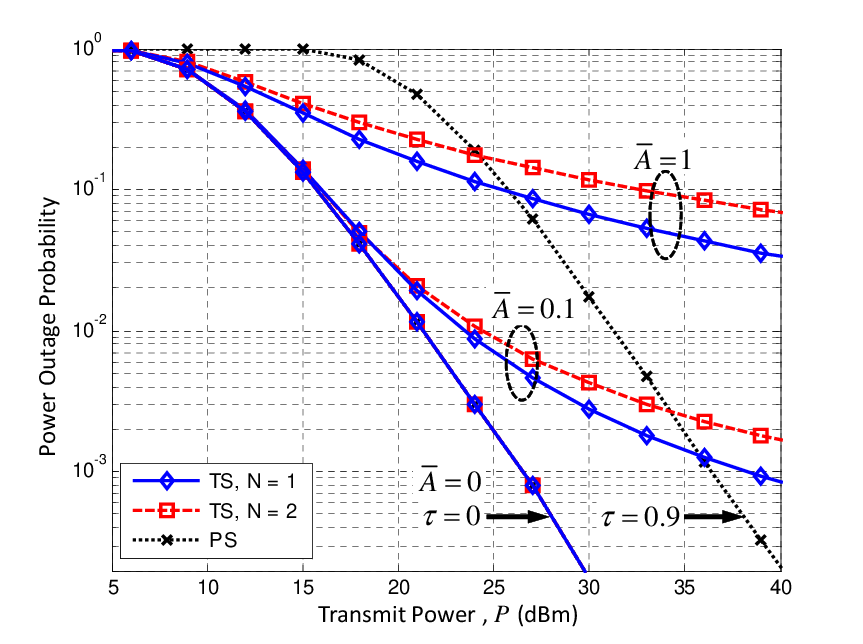}
      \caption{Power outage probability with $N_t = 2$ and $\hat Q = 1\mu$W.}
      \label{Fig_EnergyOutage_vs_P_EnergyDiversity}
   \end{figure}

   On the other hand, in the i.i.d. Rayleigh fading MISO channels, the power outage probability of PS with $P \to \infty$ can be obtained as $p_{Q,out}^{({\rm{P}})} = {\left( {\frac{\hat Q}{{(1 - \tau) \theta P }}} \right)^{ N_t}}$, $0 \le \tau < 1$; thus, from (\ref{Eq_Energy_TimeSharing}) and the fact that ${F_H}\left( h \right) \approx {h^{ N_t}}$ as $h \to 0$, we obtain the power diversity order as $d_Q^{({\rm{P}})} = N_t$, $0 \le \tau < 1$.

   Fig. \ref{Fig_EnergyOutage_vs_P_EnergyDiversity} shows the power outage probabilities of TS and PS versus the transmit power $P$ in dBm when $N_t = 2$ and $\hat Q = 1\mu$W with the same setup as for Fig. \ref{Fig_HarvestedEnergy_vs_H}, i.e., $\theta = 10^{-4}$. It is observed that the smallest power outage probabilities are achieved by TS with $\bar A = 0$ or equivalently PS with $\tau = 0$. When $\bar A >0$, $p_{Q,out}^{({\rm{T}})}$ for TS is observed to decrease slower with increasing $P$ than $p_{Q,out}^{({\rm{P}})}$ for PS, since $d_Q^{({\rm{T}})} = 0$, $\bar A > 0$ for TS while $d_Q^{({\rm{P}})} = N_t$, $0 \le \tau < 1$, for PS. Furthermore, it is also observed that $p_{Q,out}^{({\rm{T}})}$ decreases slower with increasing $P$ as $\bar A$ and/or $N$ increases, which is consistent with (\ref{Eq_EnergyDiversity_TBS}).

   \subsection{Numerical Results}\label{NumericalExample1}
   In this subsection, we compare the rate-energy performance of TS and PS for a practical SWIPT system setup and $N_t = 2$ with the same channel setup as for Figs. \ref{Fig_HarvestedEnergy_vs_H} and \ref{Fig_AchievableRate_vs_H}. It is further assumed that energy conversion efficiency is set to be $\zeta = 0.5$ to reflect practical power harvesting efficiency.

   \begin{figure}
      \centering
      \includegraphics[width=0.6\columnwidth]{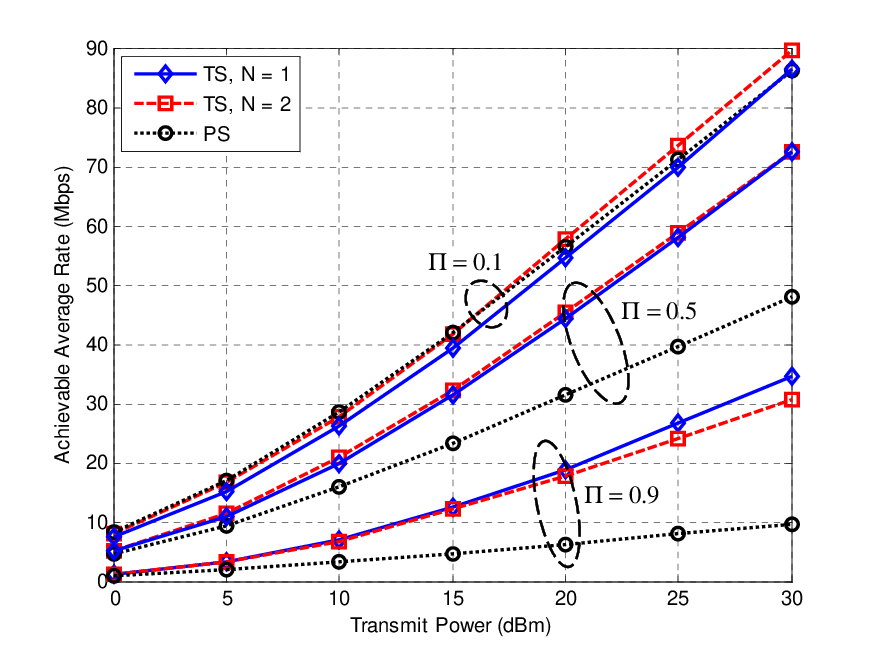}
      \caption{Comparison of the achievable average rate with $N_t = 2$ and $\Pi = 0.1$, $0.5$, and $0.9$.}
      \label{Fig_AvgRate_vs_P_Different_N}
   \end{figure}

   As inferred from Theorem \ref{Theorem_Optimality_N}, the asymptotic rate scaling factor of TS with $N = 1$ is the largest among all values of $N$ and is also larger than that for PS for a given power scaling factor as $P \to \infty$. However, it does not imply that the largest achievable average rate is always attained for a given average harvested power when $P$ is finite. Therefore, it is necessary to compare the achievable average rates for PS and TS with finite values of $P$. Fig. \ref{Fig_AvgRate_vs_P_Different_N} shows the achievable average rates for PS and TS versus transmit power in dBm under the same power scaling factor $\Pi = {\Pi ^{({\rm{T}})}}\left( {N,\bar A} \right) = {\Pi ^{({\rm{P}})}}\left( {\tau} \right)$, i.e., the same average power harvesting requirement $\bar Q = \zeta \, \theta P \, \Pi$ (e.g., $\bar Q = 45 \, \mu{\rm{W}}$ with $\zeta = 0.5$, $\theta = -40$dB, $\Pi = 0.9$ and $P$ = 30dBm). When $\Pi = 0.9$, the benefit from a larger rate scaling factor is clearly observed for TS with $N = 1$, since it achieves the largest average information rate. When $\Pi = 0.5$, the achievable average rates for TS are similar with $N = 1$ and $2$, but still grow faster with the transmit power than that for PS. When $\Pi = 0.1$, the gaps between rate scaling factors of different schemes are small (cf. Fig. \ref{Fig_AvgEnergy_vs_PreLogFactor}) and as a result their achievable average rates become similar.

   It is worth noting that one typical application scenario of the SWIPT is wireless sensor network, for which the power consumption at each sensor node is in general limited to $5$-$20\,\mu{\rm{W}}$. As observed in Fig. \ref{Fig_AvgRate_vs_P_Different_N}, with 30dBm (or 1W) transmit power, the amount of average harvested power at each receiver is $5$-$45\,\mu{\rm{W}}$ with a practical energy harvesting efficiency of $50\%$, which satisfies the power requirement of practical sensors. Furthermore, the received power can always be increased if transmit power is increased and/or the transmission distance is decreased, to meet higher power requirement of other wireless applications.

   Next, Fig. \ref{Fig_AvgRate_EnergyOutage_Region} shows the trade-offs between the achievable average rate and power non-outage probabilities, i.e., $1 - p_{Q,\,out}$, of TS and PS schemes under the same per-block harvested power requirements $\hat Q = 25\,\mu\rm{W}$ or $45\,\mu \rm{W}$ when the transmit power is set to be $30$dBm. It is observed that the minimum power outage probability of TS is attained by $N = 1$ when the achievable average rate is small, but by $N = 2$ when the achievable average rate is larger, while TS with both $N = 1$ and $2$ outperforms PS.

   \begin{figure}
      \centering
      \includegraphics[width=0.6\columnwidth]{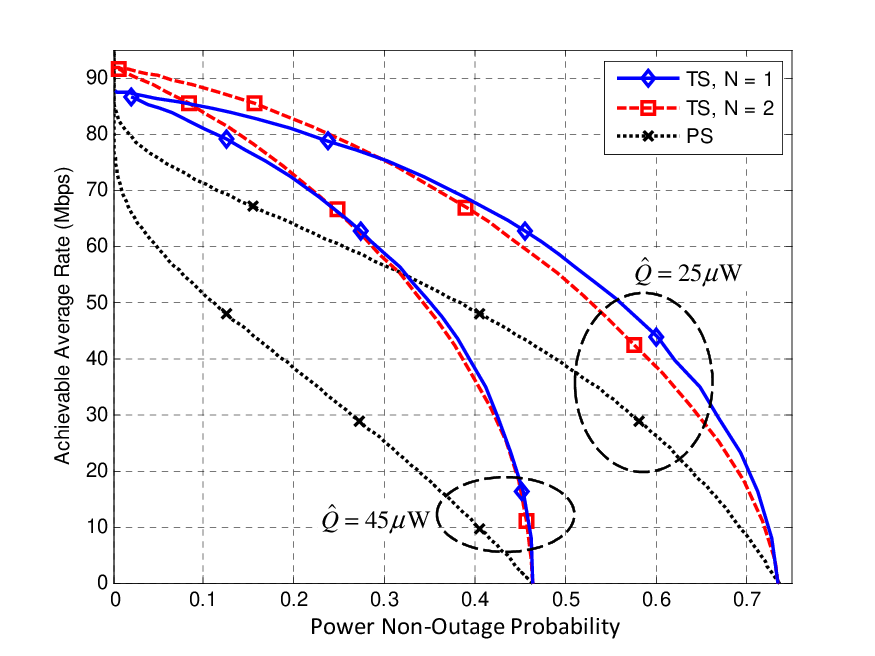}
      \caption{Trade-off between achievable average rate and power non-outage probability with $N_t = 2$, $30$dBm transmit power, and per-block harvested power requirement $\hat Q = 25$ or $45$ $\mu \rm{W}$.}
      \label{Fig_AvgRate_EnergyOutage_Region}
   \end{figure}

   \begin{remark}\label{Remark_Rate_Performance}
      Employing random beamforming at the transmitter requires additional complexity. However, from the above results, it is inferred that the achievable average rate is maximized by using only one single RB, i.e., $N = 1$, when the transmit power is asymptotically large or at finite transmit power when more harvested power is required (which is of more practical interest). In addition, TS with one single RB also optimizes power outage performance when transmit power is finite and large harvested power is required in each transmission block. Therefore, TS with one single RB in general can achieve the optimal WET efficiency and/or reliability with a given WIT rate requirement, thus yielding an appealing low-complexity implementation for practical systems.
   \end{remark}

   Finally, we investigate the overall network throughput in the multicast SWIPT system with the proposed TS scheme, which is defined as
   \begin{equation}\label{Eq_NW_Throughput}
      {{C} \buildrel \Delta \over =  \sum\limits_{i = 1}^K {\left( {1 - {p_{R,out}}\left( i \right)} \right) \bar R} ,}
   \end{equation}
   with $K$, ${{p_{R,out}}\left( i \right)}$, and $\bar R$ denoting the number of users in the network, the rate outage probability of the $i$th Rx, and the common information rate, respectively. It is worth noting that each Rx can adjust its threshold $\bar A_i$, $i = 1, \cdots, K$, according to the individual channel condition and rate requirement assuming that Rxs move slowly with a sufficiently large channel coherence time; therefore, rate outage of the $i$th user occurs when its average achievable rate cannot meet the rate target $\bar R$ even with $\bar A_i = \infty$, i.e., when all the received sub-blocks are allocated to ID mode for a given $\theta$. Accordingly, ${p_{R,out}}\left( i \right)$ is given by
   \begin{equation}\label{Eq_RateOutage}
      {{p_{R,out}}\left( i \right) = \Pr \left( {\bar R_i^{\left( {\rm{T}} \right)}}\left( {N,\infty} \right) < \bar R \right),}
   \end{equation}
   where ${\bar R_i^{\left( {\rm{T}} \right)}}\left( {N,\bar A} \right) = {{\mathbb E_H}[ {{R_i^{\left( {\rm{T}} \right)}}\left( {h,N,\bar A} \right)} ]}$ with ${{R_i^{\left( {\rm{T}} \right)}}\left( {h,N,\bar A} \right)}$ denoting the achievable rate of the $i$th Rx for given $N$ and $\bar A_i$ in a block with the normalized channel power $h$, which is given by (\ref{Eq_Rate_TBS2}). Note that for each Rx, $\theta_i$, $i = 1, \,\, \cdots, \,\, K$, can be modeled as $\theta_i = \theta_{L,i}\theta_{S,i}$ where $\theta_{L,i}$ and $\theta_{S,i}$ denote signal power attenuation due to distance-dependent pathloss and shadowing, respectively. Assuming fixed Rx locations, therefore, ${p_{R,out}}\left( i \right)$ should be measured according to the variation of $\theta_{S,i}$ in this case. Fig. \ref{Fig_NW_Throughput} shows the trade-offs between the network throughput $C$ defined in (\ref{Eq_NW_Throughput}) versus the average sum harvested power by all Rxs, denoted by $\bar Q$, under the same channel setup as for Figs. \ref{Fig_HarvestedEnergy_vs_H} and \ref{Fig_AchievableRate_vs_H}, with $K = 10$, $N = 1$, and $P = 30$dBm. The distance between the Tx and the $i$th Rx, denoted by $D_i$, is assumed to be uniformly distributed within $3{\rm{m}} \le D_i \le 10{\rm{m}}$, $i = 1, \,\, \cdots, \,\, K$. It is also assumed that $\theta_{L,i} = C_0 D_i^{-\alpha}$ with $C_0 = -20$dB denoting the pathloss at the reference distance $1$m and $\alpha = 3$ denoting the pathloss exponent. Assuming indoor shadowing, $\theta_{S,i}$ is drawn from lognormal distribution with standard deviation given by $3.72\,{\rm{dB}}$ \cite{Liberti}. Furthermore, each Rx is assumed to set $\bar A$ such that ${\bar R_i^{\left( {\rm{T}} \right)}}\left( {N,\bar A} \right) = \bar R$ if ${\bar R_i^{\left( {\rm{T}} \right)}}\left( {N,\infty} \right) \ge \bar R$, but set $\bar A = 0$, i.e. all the received power is used for power harvesting, otherwise. It is observed that the maximum throughput in the network is $C^* = 46.8$Mpbs with average harvested sum power $\bar Q^* = 424\mu$W.

   \begin{figure}
      \centering
      \includegraphics[width=0.6\columnwidth]{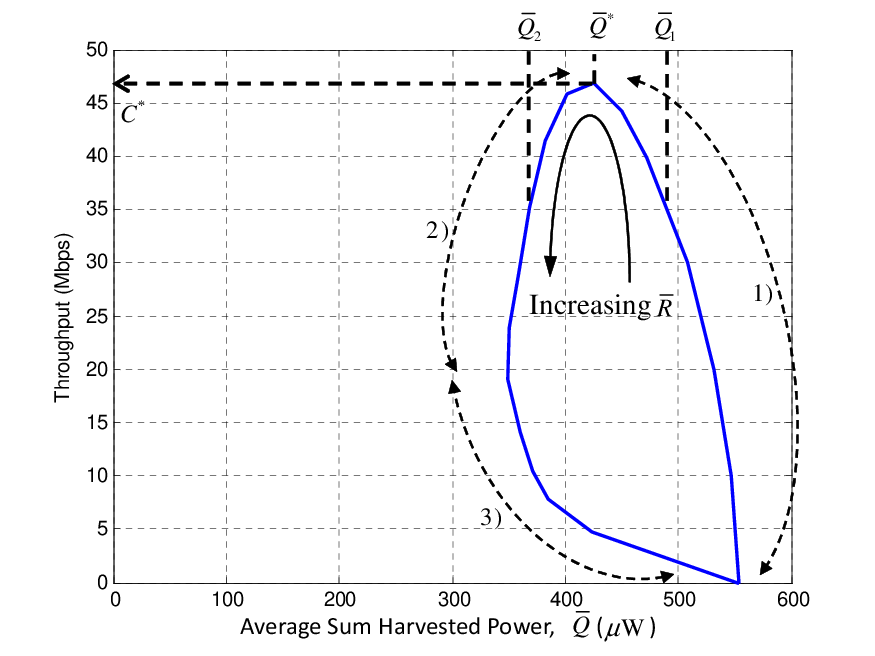}
      \caption{Trade-off between network throughput and average sum harvested power with $N_t = 2$, $N = 1$, and $30$dBm transmit power.}
      \label{Fig_NW_Throughput}
   \end{figure}

   In addition, the trade-offs shown in Fig. \ref{Fig_NW_Throughput} can be categorized into three regimes, as denoted by $1)$, $2)$, and $3)$ in the figure. When $\bar R$ is small, i.e., in the regime denoted by $1)$, $C$ increases with $\bar R$ since ${p_{R,out}}\left( i \right)$ is small. In this regime, each Rx sets larger $\bar A_i$ with increasing $\bar R$ to meet the rate target, and thus the harvested sum power decreases accordingly. When $\bar R$ is larger than a certain threshold, i.e., in the regime denoted by $2)$, $C$ decreases with increasing $\bar R$ since the number of Rx with rate outage increases. In this regime, $\bar Q$ also decreases with increasing $\bar R$ since Rxs with large $\theta_i$'s still set larger $\bar A_i$ with increasing $\bar R$ and their harvested power decreases. Finally, when $\bar R$ further increases, i.e., in the regime denoted by $3)$, $C$ decreases with increasing $\bar R$ whereas $\bar Q$ increases with $\bar R$. This is because most of Rxs in the network experience rate outage and thus only harvest power. When $\bar R \to \infty$, therefore, $C \to 0$ and $\bar Q$ becomes equivalent to that without WIT and with WET only, i.e., $\bar R = 0$ with $\bar A = 0$. Therefore, for a given throughput $C < C^*$, there are two possible values of average sum harvested power (e.g., $\bar Q_1 = 490\mu$W and $\bar Q_2 = 368\mu$W with $C = 20$Mbps), and thereby we can choose larger value of average sum harvested power for a given throughput (e.g., $\bar Q_1$ for the aforementioned example).

\section{Alternative Random Beam Designs}\label{Sec:OtherBeams}
It is worth noting that TS with Gaussian random beams (referred to as TS-G), as considered in the preceding sections, may not be practically favorable due to the fact that Gaussian random beams (GRBs) will cause large transmit power at certain sub-blocks. Instead, artificial channel fading within each transmission block of each Rx can be generated by employing non-Gaussian random beams with constant transmit power for TS. In this section, we investigate the performance of TS with two alternative RBs other than GRB, such that the average transmit power remains constant within each transmission block, which are given next.

   \subsection{Unitary Random Beams (URBs)}
   In this case, $N$ unitary random vectors obtained from the isotropic distribution \cite{Hassibi} are independently employed for the $N$ random beams at the $k$th sub-block, i.e., $\boldsymbol{\phi}_n\left( k \right)$, $1 \le n \le N$, $\forall k$.

   With URBs, it is in general difficult to obtain the closed-form expressions for the PDF and CDF of the received channel power $A \left( k \right)$ at each sub-block conditioned on $H = h$. However, if we consider the special case of $N_t = 2$ and $N = 1$, it is known that with URBs $A \left( k \right)$ is uniformly distributed within $[0, 2h]$. Thus, given a threshold $\bar A \ge 0$, the amount of harvested power in each block with URBs can be obtained using (\ref{Eq_Energy_TBS2}) as
   \begin{equation}\label{Eq_URB_Energy}
      {Q^{(\rm{U})}\left( {h,\bar A} \right) = \left\{ {\begin{array}{*{20}{c}}
      {\theta P(h - \frac{{\bar A}^2}{4h})}  \\
      0  \\
      \end{array}} \right.\begin{array}{*{20}{c}}
      {,\,\,\,\, 0 \le \bar A \le 2h\,\,\,}  \\
      {,\,\,\,\,\,\,\,\, \bar A > 2h. \,\,\,\,\,\,\,\,}  \\
      \end{array}}
   \end{equation}

   In the i.i.d. fading MISO channel, the average harvested power for TS with URBs (referred to as TS-U) given a fixed threshold $\bar A \ge 0$ is obtained as ${\bar Q^{({\rm{U}})}}\left( {\bar A} \right) = \int_0^\infty  {{Q^{({\rm{U}})}}\left( {h,\bar A} \right){f_H}\left( h \right)dh}$, where ${f_H}\left( h \right) = 4h{e^{ - 2h}}$ is given by (\ref{Eq_pdf_H}) for $N_t = 2$. It is worth noting that in the special case of $N_t = 2$ and $N = 1$, the unconditional distribution of $A \left( k \right)$ with URBs after averaging over the fading channels can be shown to be the exponential distribution, where the unconditional PDF is given by ${q_A^{(\rm{U})}}\left( a \right) = {e^{ - a}}$. Therefore, $\bar Q^{(\rm{U})} \left( {\bar A} \right)$ can be alternatively obtained as
   \[
      {{\bar Q}^{({\rm{U}})}}\left( {\bar A} \right) = \int_{\bar A}^\infty  {\theta P a{q_A^{(\rm{U})}}\left( a \right)da}
   \]
   \begin{equation}\label{Eq_AvgEnergy_TS_U}
      {\,\,\,\,\,\, = \theta P \Gamma \left( {2,\bar A} \right), }
   \end{equation}
   which is equivalent to ${Q^{(\rm{T})}}\left( {1,1,\bar A} \right)$ for TS-G given by (\ref{Eq_Energy_TBS3}) with $N = 1$ and $h = 1$. Similarly, given $\bar A \ge 0$, the achievable average transmission rate for TS-U is obtained as
   \[
      {{\bar R}^{({\rm{U}})}}\left( {\bar A} \right) = \int_0^{\bar A} {{{\log }_2}\left( {1 + \frac{\theta P a}{\sigma^2}} \right){q_A^{(\rm{U})}}\left( a \right)da} \,\,\,\,\,\,\,\,\,\,\,\,\,\,\,\,\,\,\,\,\,\,\,\,\,\,\,
   \]
   \begin{equation}\label{Eq_AvgRate_TS_U}
      {\,\,\,\,\,\,\,\,\,\,\,\,\,\,\,\,\,\,\,\,\,\,\,\,\,\,\,\,\,\,\,\,\,\,\,\,\,\,\,\,\,\,\,\,\,\,\,\,\,\,\,\,\,\,\,\,\,\,\,\,\,\,\,\,\,\,\,\,\,\,\,\,\,  = \frac{{{e^{ - \frac{1}{{\eta}}}}}}{{\ln 2}}\left( {{E_1}\left( {\frac{\sigma^2}{\theta P}} \right) - {E_1}\left( \bar A + {\frac{\sigma^2}{{\theta P}}} \right)} \right) - {e^{ - \bar A}}{\log _2}\left( {1 + \frac{\theta P \bar A}{\sigma^2}} \right),}
   \end{equation}
   with ${E_n}\left( z \right) = \int_1^\infty  {{e^{ - zt}}{t^{ - n}}dt}$ denoting the exponential integral function for integer $n \ge 0$, which is also equivalent to ${R^{({\rm{T}})}}\left( {1,1,\bar A} \right) $ for TS-G given by (\ref{Eq_Rate_TBS2}) with $h = 1$ and $N = 1$. In addition, given fixed $\bar A \ge 0$ and per-block power harvesting requirement $\hat Q > 0$, the power outage probability of TS-U in the i.i.d. Rayleigh fading MISO channel with $N_t = 2$ and $N = 1$ can be obtained from (\ref{Eq_URB_Energy}) as
   \begin{equation}\label{Eq_URB_EnergyOutage}
      {p_{Q,out}^{({\rm{U}})} = {F_H}\left( {\frac{{{{\hat Q}^2} + \sqrt {{{\hat Q}^2} + {{\theta^2}{P^2}{\bar A}^2}} }}{{2\theta P}}} \right),}
   \end{equation}
   where ${F_H}\left( h \right)$ is given by (\ref{Eq_CDF_H}) for $N_t = 2$. From (\ref{Eq_EnergyDiversity}) and (\ref{Eq_URB_EnergyOutage}), it is easily verified that TS-U in the case of $N_t = 2$ and $N = 1$ has the power diversity order of $0$, the same as TS-G.

   \subsection{Binary Random Beams (BRBs)}
   In this case, a random subset of $N$ out of $N_t$ transmit antennas at Tx, $1 \le N \le N_t$, are selected to transmit at each sub-block, which is equivalent to selecting ${\bf{\Phi}} \left( k \right) = [ {{\boldsymbol{\phi} _1}( k )\,\,{\boldsymbol{\phi} _2}( k )\,\, \cdots \,\,{\boldsymbol{\phi} _N}( k )} ] \in {{\mathbb R}^{{N_t} \times N}}$, $\forall k$, where
   ${\boldsymbol{\phi} _n}( k ) = [{\phi _{n,1}}\left( k \right) \,\,\, {\phi _{n,2}}\left( k \right) \,\,\, \cdots {\phi _{n,N_t}}\left( k \right)]^T$, $1 \le n \le N$, with ${\phi _{n,i}}\left( k \right) \in \left\{ {0,1} \right\}$, $1 \le i \le N_t$, such that ${\left\| {\boldsymbol{\phi} _n}( k ) \right\|^2} = 1$ and ${{\boldsymbol{\phi} _{m}^{T}}{{\left( {k} \right)}}{\boldsymbol{\phi} _n}\left( k \right)} = 0$, $n \ne m$. We assume that all the subsets of the selected antennas are equally probable.

   Consider the special case of $N_t = 2$ and $N = 1$. Denote ${\bf{h}} = {[{h_1}\,\,\,{h_2}]^T}$, $V = \max ( \, {{{\left| {{h_1}} \right|}^2},{{\left| {{h_2}} \right|}^2}} )$, and $W = \min ( \, {{{\left| {{h_1}} \right|}^2},{{\left| {{h_2}} \right|}^2}} )$. Note that in this case the received channel power at each sub-block is either $A \left( k \right) = V$ or $A \left( k \right) = W$, each of which occurs with a probability of $1/2$. Thus, given $V = v$, $W = w$, and a fixed threshold $\bar A \ge 0$, the amount of harvested power in each block with BRBs is obtained using (\ref{Eq_Energy_TBS2}) as
   \begin{equation}\label{Eq_AS_Energy}
      {Q^{(\rm{B})}\left( {v,w,\bar A} \right) = \left\{ {\begin{array}{*{20}{c}}
      \theta P\left( {v+w} \right)/2  \\
      {\theta Pv / 2}  \\
      0  \\
      \end{array}\begin{array}{*{20}{c}}
      {, \,\,\,\,\,\,\, \bar A < w \,\,\,\,\,\,\,}  \\
      {, \,\,\, w \le \bar A \le v}  \\
      {, \,\,\,\,\,\,\, \bar A > v. \,\,\,\,\,\,\,}  \\
      \end{array}} \right.}
   \end{equation}

   Similar to TS-U, in the i.i.d. fading MISO channel, it can be shown that with $N_t = 2$ and $N = 1$ the unconditional distribution of $A \left( k \right)$ with BRBs after averaging over the fading channels is the exponential distribution, where the unconditional PDF is also given by ${q_A^{(\rm{B})}}\left( a \right) = {e^{ - a}}$. Therefore, given a fixed threshold $\bar A \ge 0$, $\bar Q^{(\rm{B})} \left( {\bar A} \right) = \bar Q^{(\rm{U})} \left( {\bar A} \right)$ and $\bar R^{(\rm{B})} \left( {\bar A} \right) = \bar R^{(\rm{U})} \left( {\bar A} \right)$, where $\bar Q^{(\rm{B})} \left( {\bar A} \right)$ and $\bar R^{(\rm{B})} \left( {\bar A} \right)$ denote the average harvested power and achievable average information rate  for TS with BRBs (referred to as TS-B), respectively, and $\bar Q^{(\rm{U})} \left( {\bar A} \right)$ and $\bar R^{(\rm{U})} \left( {\bar A} \right)$ for TS-U are given by (\ref{Eq_AvgEnergy_TS_U}) and (\ref{Eq_AvgRate_TS_U}), respectively. In addition, given fixed $\bar A \ge 0$ and per-block power harvesting requirement $\hat Q > 0$, the power outage probability of TS-B in the i.i.d. Rayleigh fading MISO channel with $N_t = 2$ and $N = 1$ can be obtained from (\ref{Eq_AS_Energy}) as (see Appendix \ref{App_Derivation_P_E_out_B} for the detailed derivation)
   \[
      p_{Q,out}^{({\rm{B}})} = {\left( {1 - {e^{ - \bar A}}} \right)^2} + {\bf{1}}\left( {\bar A < 2D} \right) \cdot 2{e^{ - 2\left( {\bar A + D} \right)}}\left( { - 1 + {e^{\bar A}}} \right)\left( { - {e^{\bar A}} + {e^{2D}}} \right) \,\,\,\,\,\,\,\,\,\,\,\,\,\,\,\,\,
   \]
   \begin{equation}\label{Eq_BRB_EnergyOutage}
      {\,\,\,\,\,\,\,\,\,\,\,\,\,\,\,\,\,\,\,\,\,\,\,\,\,\,\,\,\,\,\,\,\,\,\,\,\,\,\,\,\,\, + {\bf{1}}\left( {\bar A < D} \right)\left( {{e^{ - 2\left( {\bar A + D} \right)}}{{\left( {{e^{\bar A}} - {e^D}} \right)}^2} + {e^{ - \bar A - 2D}}\left( {\left( { - 1 + \bar A - D} \right){e^{\bar A}} + {e^D}} \right)} \right),}
   \end{equation}
   where $D = {\hat Q}/({\theta P})$, and ${\bf{1}}\left( {x < y} \right)$ denotes the indicator function given by
   \[{\bf{1}}\left( {x < y} \right) = \left\{ {\begin{array}{*{20}{c}}
     1  \\
     0  \\
     \end{array}\begin{array}{*{20}{c}}
     {,\,\,\,\, {\rm{if}} \,\, x < y \,\,\,\,\,}  \\
     {,\,\,\,{\rm{otherwise}}.}  \\
   \end{array}} \right.\]
   From (\ref{Eq_BRB_EnergyOutage}), it is shown that both ${\bf{1}}\left( {\bar A < 2D} \right) = 0$ and ${\bf{1}}\left( {\bar A < D} \right) = 0$ if $P > 2\hat Q / (\theta \bar A)$, and thus $p_{Q,out}^{({\rm{B}})} = { ( {1 - {e^{ - \bar A}}} )^2}$. Therefore, TS-B in the case of $N_t = 2$ and $N = 1$ also has the power diversity order of $0$, the same as TS-G and TS-U.

\begin{figure}
   \centering
   \includegraphics[width=0.6\columnwidth]{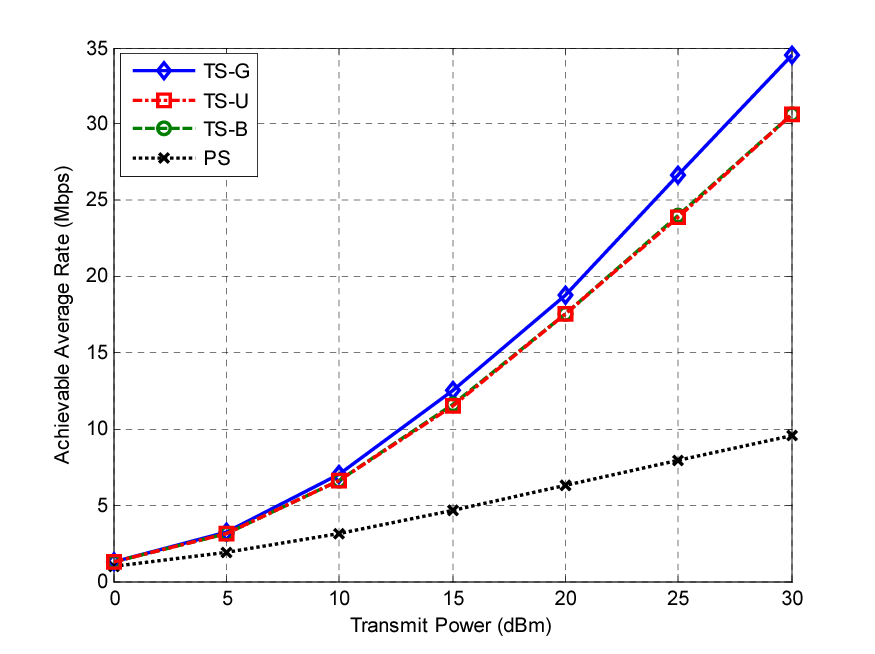}
   \caption{Comparison of the achievable average rate for different RB designs with $N_t = 2$, $N = 1$, and $\Pi = 0.9$.}
   \label{Fig_AvgRate_vs_P_DifferentBeams}
\end{figure}

Fig. \ref{Fig_AvgRate_vs_P_DifferentBeams} shows the achievable average rates of TS-G, TS-U, TS-B, and PS versus transmit power in dBm for the same setup as for Fig. \ref{Fig_AvgRate_vs_P_Different_N}, under the same average power harvesting requirement with $\Pi = 0.9$. It is observed that the achievable average information rates of TS-U and TS-B are the same, which is as expected for the considered case here of $N_t = 2$ and $N = 1$. It is also observed that the achievable average rates of TS-U and TS-B are larger than that of PS, but smaller than that of TS-G. This result is originated from the fact that the artificial channel fading generated by URBs or BRBs in this case is less substantial over time than that by GRBs, due to the limitation of constant average transmit power over sub-blocks with URBs or BRBs.

\begin{figure}
   \centering
   \includegraphics[width=0.6\columnwidth]{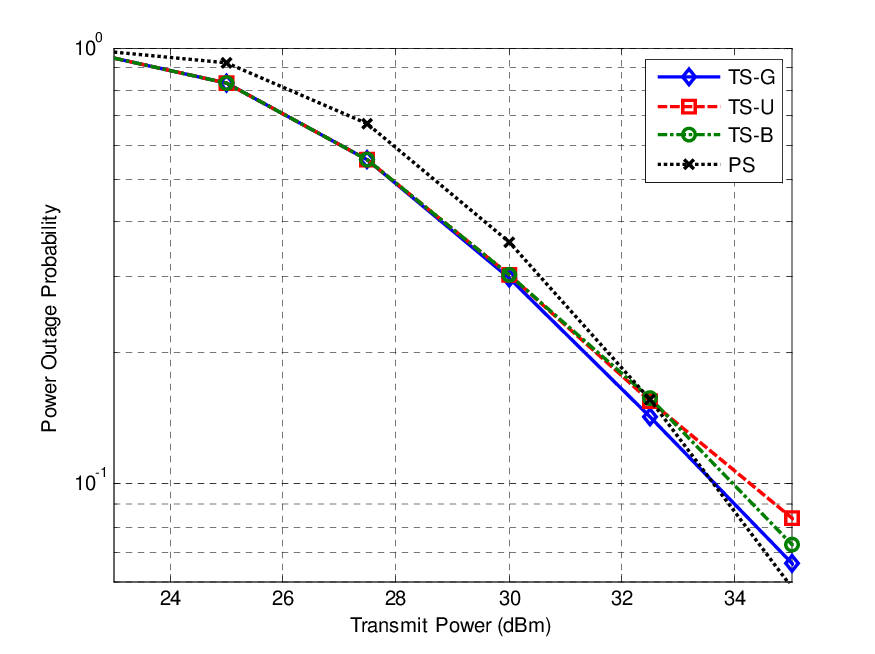}
   \caption{Comparison of power outage probability for different RB designs with $N_t = 2$, $N = 1$, $\bar R = 2$ bps/Hz, and $\hat Q = 25\mu \rm{W}$.}
   \label{Fig_EnergyOutage_vs_P_FixedRbar}
\end{figure}

Fig. \ref{Fig_EnergyOutage_vs_P_FixedRbar} shows the power outage probabilities of TS-U and TS-B versus transmit power in dBm for the same setup as for Fig. \ref{Fig_AchievableRate_vs_H}, when $\bar R = 2$ bps/Hz and $\hat Q = 25$ $\mu \rm{W}$, as compared to that of PS and TS-G. Among TS schemes, it is observed that the power outage probability of TS-G is the smallest. The power outage probability of TS-U is observed to be similar to that of TS-G when transmit power is small, but becomes larger than that of TS-G when transmit power is larger than $30$dBm. The power outage probability of TS-B is observed to be between those of TS-U and of TS-G. It is also observed that the power outage probability of PS is larger than those of all TS schemes when transmit power is small, but is the smallest when transmit power is larger than $33$dBm.

\section{Conclusion} \label{Sec:Conclusion}
This paper has studied a novel receiver mode switching scheme for the MISO multicast SWIPT system when the channel is only known at the receiver, but unknown at the transmitter. The proposed scheme exploits the benefit of opportunistic energy harvesting over artificial channel fading induced by employing multi-antenna random beamforming at the transmitter. By investigating the achievable average information rate, average harvested power/power outage probability, and their various trade-offs, it is revealed that the proposed scheme yields better power and information transfer performance than the reference scheme of periodic switching without transmit random beamforming when the harvested power requirement is sufficiently large. Particularly, employing one single random beam for the proposed scheme is proved to achieve the asymptotically optimal trade-off between the average information rate and average harvested power when transmit power goes to infinity. Moreover, it is shown by simulations that the best trade-offs between average information rate and average harvested power/power outage probability are also achieved by the proposed scheme employing one single random beam for large power harvesting targets of most practical interests, even with finite transmit power.

\appendices

   \section{Derivations of ${R^{({\rm{T}})}}\left( {h,1,\bar A} \right)$ and ${R^{({\rm{T}})}}\left( {h,2,\bar A} \right)$}\label{App_Rate_N1_Derivation}
   From (\ref{Eq_TBS_pdf}) with $N = 1$ and (\ref{Eq_Rate_TBS2}), ${R^{({\rm{T}})}}\left( {h,1,\bar A} \right)$ can be expressed as
   \begin{equation}\label{App_Rate_N1_1}
      {{R^{({\rm{T}})}}\left( {h,1,\bar A} \right) = \frac{1}{{\ln 2}}\int_0^{\bar A} {\ln \left( {1 + \tilde P a} \right)\frac{1}{h}{e^{ - \frac{a}{h}}}} da}
   \end{equation}
   \begin{equation}\label{App_Rate_N1_2}
      {\,\,\,\,\,\,\,\,\,\,\,\,\,\,\,\,\,\,\,\,\,\,\,\,\,\,\,\,\,\,\,\,\,\,\,\,\,\,\,\,\,\,\,\,\,\,\,\,\,\,\,\,\,\,\,\,\,\,\,\,\,\,\,\,\,\,\,\,\,\,\,\,\,\,\,\,\,\,\,\,\,\, = \frac{1}{{\ln 2}}\left( {{e^{ - \frac{{\bar A}}{h}}}\ln \left( {1 + \tilde P \bar A} \right) + \int_0^{\bar A} {\frac{\tilde P}{{1 + \tilde P a}}{e^{ - \frac{a}{h}}}} da} \right),}
   \end{equation}
   where $\tilde P = \frac{\theta P}{\sigma^2}$ and (\ref{App_Rate_N1_2}) is obtained from integrating (\ref{App_Rate_N1_1}) by part. By changing a variable as $x = 1 + \tilde P a$, the integral term in (\ref{App_Rate_N1_2}) can be obtained as
   \[
      \int_0^{\bar A} {\frac{\tilde P}{{1 + \tilde P a}}{e^{ - \frac{a}{h}}}} da = \int_1^{1 + \tilde P\bar A} {\frac{1}{x}{e^{ - \frac{{x - 1}}{{\tilde Ph}}}}} dx \,\,\,\,\,\,\,\,\,\,\,\,\,\,\,\,\,\,\,\,\,\,\,\,\,\,\,\,\,\,\,\,\,\,\,\,\,\,\,\,
   \]
   \begin{equation}\label{App_Rate_N1_3}
      {\,\,\,\,\,\,\,\,\,\,\,\,\,\,\,\,\,\,\,\,\,\,\,\,\,\,\,\,\,\,\,\,\,\,\,\,\,\,\,\,\,\,\,\,\,\,\,\,\,\,\,\,\,\,\,\,\,\,\,\,\,\,\,\,\,\,\, = {e^{ \frac{1}{{\tilde P h}}}}\left( {\int_1^\infty  {\frac{1}{x}{e^{ - \frac{x}{{\tilde P h}}}}} dx - \int_{1 + \tilde P \bar A}^\infty  {\frac{1}{x}{e^{ - \frac{x}{{\tilde P h}}}}} dx} \right)}
   \end{equation}
   \begin{equation}\label{App_Rate_N1_4}
      { \,\,\,\,\,\,\,\,\,\,\,\,\,\,\,\,\,\,\,\,\,\,\,\,\,\,\,\,\,\,\,\,\,\,\,\,\,\,\,\,\,\,\,\,\,\,\,\,\,\,\,\,\,\,\,\,\, = {e^{ \frac{1}{{\tilde P h}}}}\left( {{E_1}\left( {\frac{1}{{\tilde P h}}} \right) - {E_1}\left( {\frac{{1 + \tilde P \bar A}}{{\tilde P h}}} \right)} \right),}
   \end{equation}
   where, for integer $n \ge 0$, ${E_n}\left( z \right) = \int_1^\infty  {{e^{ - zt}}{t^{ - n}}dt}$ denotes the exponential integral function; (\ref{App_Rate_N1_4}) is obtained by a change of variable as $y = \left( {1 + \tilde P \bar A} \right)x$ for the second integral term in (\ref{App_Rate_N1_3}). From (\ref{App_Rate_N1_2}), (\ref{App_Rate_N1_4}), and $\tilde P = \frac{\theta P}{\sigma^2}$, ${R^{({\rm{T}})}}\left( {h,1,\bar A} \right)$ is obtained as
   \begin{equation}\label{Eq_Rate_TBS_N_1}
      {{R^{({\rm{T}})}}\left( {h,1,\bar A} \right) = \frac{{{e^{  \frac{\sigma^2}{{\theta P h}}}}}}{{\ln 2}}\left( {{E_1}\left( {\frac{\sigma^2}{{\theta P h}}} \right) - {E_1}\left( {\frac{{\bar A}}{{h}} + \frac{\sigma^2}{\theta Ph}} \right)} \right) - {e^{ - \frac{{\bar A}}{h}}}{\log _2}\left( {1 + \frac{\theta P\bar A}{\sigma^2}} \right).}
   \end{equation}

   When $N = 2$, ${R^{({\rm{T}})}}\left( {h,2,\bar A} \right)$ can be derived similarly by integrating (\ref{Eq_Rate_TBS2}) by part. In this operation, it is necessary to apply differentiation of the incomplete Gamma function given by \cite{Geddes}
   \begin{equation}\label{App_Meijer}
      {\frac{\partial }{{\partial N}}\Gamma \left( {N,x } \right) = \Gamma \left( {N,x } \right)\ln x  + x G_{2,3}^{3,0}\left( {\left. {\begin{array}{*{20}{c}}
      {0,0}  \\
      {N - 1, - 1, - 1}  \\
      \end{array}} \right|x } \right),}
   \end{equation}
   where $G_{p,q}^{m,n}\left( {\left. {\begin{array}{*{20}{c}} {{a_1},\,\, \cdots ,\,\,{a_p}}  \\ {{b_1},\,\, \cdots ,\,\,{b_q}}  \\ \end{array}} \right|z} \right)$ denotes the Meijer-G function, defined as \cite[9.301]{TableOfIntegral}
   \begin{equation}\label{Meijer}
      {G_{p,q}^{m,n}\left( {\left. {\begin{array}{*{20}{c}}
      {{a_1},\,\, \cdots ,\,\,{a_p}}  \\
      {{b_1},\,\, \cdots ,\,\,{b_q}}  \\
      \end{array}} \right|z} \right) = \frac{1}{{2\pi i}}\int_L {\frac{{\prod\limits_{j = 1}^m {\Gamma \left( {{b_j} - s} \right)} \prod\limits_{k = 1}^n {\Gamma \left( {1 - {a_k} + s} \right)} }}{{\prod\limits_{j = m + 1}^q {\Gamma \left( {1 - {b_j} + s} \right)} \prod\limits_{k = n + 1}^p {\Gamma \left( {{a_k} - s} \right)} }}} {z^s}ds,}
   \end{equation}
   with $\int_L {}$ denoting the Barres integral. By the definition of the Meijer-G function, the last term in (\ref{App_Meijer}) can be represented by
   \clearpage
   \begin{equation}\label{App_Meijer_Simplify1}
      {xG_{2,3}^{3,0}\left( {\left. {\begin{array}{*{20}{c}}
      {0,0}  \\
      {N - 1, - 1, - 1}  \\
      \end{array}} \right|x} \right) = \frac{1}{{2\pi i}}\int_L {\frac{{\Gamma \left( {N - 1 - s} \right)\Gamma \left( { - 1 - s} \right)\Gamma \left( { - 1 - s} \right)}}{{\Gamma \left( { - s} \right)\Gamma \left( { - s} \right)}}{x^{s + 1}}ds}}
   \end{equation}
   \begin{equation}\label{App_Meijer_Simplify3}
      { = G_{2,3}^{3,0}\left( {\left. {\begin{array}{*{20}{c}}
      {1,1}  \\
      {0,0,N}  \\
      \end{array}} \right|x} \right), \,\,\,\,}
   \end{equation}
   where (\ref{App_Meijer_Simplify3}) is achieved by a change of variable as $t = s+1$ in (\ref{App_Meijer_Simplify1}). By applying (\ref{App_Meijer})-(\ref{App_Meijer_Simplify3}) to the integration of (\ref{Eq_Rate_TBS2}) by part, ${R^{({\rm{T}})}}\left( {h,2,\bar A} \right)$ can be obtained as
   \[
      {R^{({\rm{T}})}}\left( {h,2,\bar A} \right) = \left( {\frac{2 \sigma^2}{{\theta P h}}{e^{ - \frac{{2\bar A}}{h}}} - {e^{\frac{2 \sigma^2}{{\theta P h}}}}\Gamma \left( {2,2 \left( \frac{{ {\bar A} }}{{h}} + \frac{\sigma^2}{\theta Ph} \right)} \right)} \right){\log _2}\left( {1 + \frac{\theta P\bar A}{\sigma^2}} \right) \,\,\,\,\,\,\,\,\,\,\,\,\,\,\,\,\,\,\,\,\,\,\,\,\,\,\,\,\,\,
   \]
   \[
      + \frac{2 \sigma^2}{{\theta Ph\ln 2}}{e^{\frac{2 \sigma^2}{{\theta Ph}}}}\left( {{E_1}\left( 2 \left( {\frac{{{\bar A}}}{{h}} + \frac{\sigma^2}{\theta Ph}} \right) \right) - {E_1}\left( {\frac{2 \sigma^2}{{\theta Ph}}} \right)} \right)
   \]
   \begin{equation}\label{Eq_Rate_TBS_N_2}
      {\,\,\,\,\,\,\,\,\,\,\,\,\,\,\,\,\,\,\,\,\,\,\,\,\,\,\,\,\,\,\,\,\,\,\,\,\,\,\,\,\,\,\,\,\,\,\,\,\,\,\,\,\,\,\,\,\,\,\,\,\,\,\,\,\, + \frac{1}{{\ln 2}}{e^{\frac{2 \sigma^2}{{\theta Ph}}}}\left( {G_{2,3}^{3,0}\left( {\left. {\begin{array}{*{20}{c}}
          {1,1}  \\
          {0,0,2}  \\
          \end{array}} \right|\frac{2\sigma^2}{{\theta Ph}}} \right) - G_{2,3}^{3,0}\left( {\left. {\begin{array}{*{20}{c}}
          {1,1}  \\
          {0,0,2}  \\
          \end{array}} \right|   2 \left( {\frac{{ {\bar A} }}{{h}} + \frac{\sigma^2}{\theta Ph} } \right)            }\right)} \right).}
   \end{equation}

   \section{Derivations of $C_0\left( h, N, \bar A \right)$ in (\ref{Eq_Rate_TBS3})}\label{App_C0_Derivation}
   From (\ref{Eq_TBS_pdf}) and (\ref{Eq_pdf_H}), $C_0\left( {h,\bar A,N} \right)$ in (\ref{Eq_Rate_TBS3}) can be expressed as
   \[ C_0 \left( {h,\bar A,N} \right) = \int_0^{\bar A} {{{\log }_2}\left( a \right){f_{A\left| H \right.}}\left( {a\left| h \right.} \right)da} \,\,\,\,\,\,\,\,\,\,\,\,\,\,\,\,\,\,\,\,\,\,\,\,\,\,\,\,\,\,\,\,\,\,\,\,\,\,\,\,\,\,\,\,\,\,\,\,\,\,\,\,\,\,\,\,\,\,\,\,\,\,\,\,\,\,\,\,\,\,\,\,\,\,\,\,\,\,\,\,\,\,\,\,\,\,\,\,\,\,\,\,\,\,\,\,\,\,\,\,\,\,\,\,\,\,\,\,\,\,\,\,\,\,\,\,\,\,\,\,\,\,\,\,\,\,\,\,\,\,\, \]
   \begin{equation}\label{App_TBS_Constant1}
      {\,\,\,\,\,\,\,\,\,\,\,\,\,\,\,\,\, = \underbrace {\frac{{{{\left( {N/h} \right)}^N}}}{{\Gamma \left( N \right)\ln 2}}\int_0^\infty  {{a^{N - 1}}{e^{ - \frac{N}{h}a}}\ln \left( a \right)da} }_{\alpha} - \underbrace {\frac{{{{\left( {N/h} \right)}^N}}}{{\Gamma \left( N \right)\ln 2}}\smallint _{\bar A}^\infty {a^{N - 1}}{e^{ - \frac{N}{h}a}}\ln \left( a \right)da}_{\beta}.}
   \end{equation}

   From \cite[4.352-1]{TableOfIntegral}, $\alpha$ can be derived as
   \begin{equation}\label{App_TBS_Constant2}
      {\alpha =  \frac{{\Psi \left( N \right)}}{{\ln 2}} + {\log _2}\frac{h}{N}.}
   \end{equation}

   Next, by changing variable as $a = \bar A x$, $\beta$ can be modified as
   \begin{equation}\label{App_TBS_Constant3}
      {\beta = \underbrace { - {{\left( {\frac{{N\bar A}}{h}} \right)}^N}\frac{{{{\log }_2}\left( {\bar A} \right)}}{{\Gamma \left( N \right)}}\int_1^\infty  {{x^{N - 1}}{e^{ - \frac{{N\bar A}}{h}x}}dx} }_{{\beta _1}} - \underbrace {{{\left( {\frac{N\bar A}{h}} \right)}^N}\frac{1}{{\Gamma \left( N \right)\ln 2}}\int_1^\infty  {{x^{N - 1}}{e^{ - \frac{{N\bar A}}{h}x}}\ln xdx} }_{{\beta _2}},}
   \end{equation}
   where, by the similar process to derive (\ref{Eq_Energy_TBS3}), $\beta_1$ is derived as
   \begin{equation}\label{App_TBS_Constant4}
      { \beta_1  =  - \frac{{\Gamma \left( {N,\frac{{N\bar A}}{{h}}} \right)}}{{\Gamma \left( N \right)}}{\log _2}\left( {\bar A} \right).}
   \end{equation}
   In addition, by changing variable as $\frac{{N\bar A}}{{h}} = \theta$, $\beta_2$ can be derived from \cite[4.358-1]{TableOfIntegral} as
   \begin{equation}\label{App_TBS_Constant5}
      {\beta_2  = \frac{{{\theta ^N}}}{{\Gamma \left( N \right)}\ln 2}\int_1^\infty  {{x^{N - 1}}{e^{ - \theta x}}\ln xdx} = \frac{\partial }{{\partial N}}\left( {{\theta ^{ - N}}\Gamma \left( {N,\theta } \right)} \right) .}
   \end{equation}
   Since $\frac{\partial }{{\partial N}}{\theta ^{ - N}} =  - {\theta ^{ - N}}\ln \theta$, $\beta_2$ in (\ref{App_TBS_Constant5}) can be obtained from (\ref{App_Meijer})-(\ref{App_Meijer_Simplify3}) as
   \[ {\beta_2 = \frac{1}{{\Gamma \left( N \right)\ln 2}}\theta G_{2,3}^{3,0}\left( {\left. {\begin{array}{*{20}{c}}
      {0,0}  \\
      {N - 1, - 1, - 1}  \\
      \end{array}} \right|\theta } \right)} \]
   \begin{equation}\label{App_TBS_Constant6}
      {= \frac{1}{{\Gamma \left( N \right)\ln 2}} G_{2,3}^{3,0}\left( {\left. {\begin{array}{*{20}{c}}
      {1,1}  \\
      {0,0,N}  \\
      \end{array}} \right|\theta } \right), \,\,\,\,\,\,\,\,\,\,\,\,\,\,\,}
   \end{equation}
   From (\ref{App_TBS_Constant1})-(\ref{App_TBS_Constant6}) and by changing variable as $\theta = {\frac{{N\bar A}}{{h}}}$ in (\ref{App_TBS_Constant6}), we arrive at (\ref{Eq_Rate_TBS3}). This completes the derivation of $C_0 \left( {h,\bar A,N} \right)$ in (\ref{Eq_Rate_TBS3}).

   \section{Proof of Proposition \ref{Proposition_Avg_Scaling}}\label{App_Proof_Proposition_Avg_Scaling}
   With $P \to \infty$, the achievable average information rate for TS is expressed from (\ref{Eq_Rate_TBS3}) as
   \[
      {\mathbb E_H}\left[ {{R^{({\rm{T}})}}\left( {h,N,\bar A} \right)} \right] = {\mathbb E_H}\left[ {F_{A\left| H \right.}^{(N)}\left( {\bar A\left| h \right.} \right)} {\log _2}\left( \frac{\theta P}{\sigma^2} \right) + {{C_0}\left( {h,N,\bar A} \right)} \right] \,\,\,\,\,\,\,\,\,\,\,\,\,\,\,\,\,\,\,\,\,\,\,\,\,\,\,\,\,\,\,\,\,\,\,\,\,\,\,\,\,\,\,
   \]
   \[
      \,\,\,\,\,\,\,\,\,\,\,\,\,\,\,\,\,\,\,\,\,\,\,\,\,\,\,\,\,\,\,\,\,\,\,\,\,\,\,\,\,\,\,\,\,\,\,\,\,\,\,\,\,\,\,\,\,\,\,\,\,\,\,\,\,\,\,\,\,\,\,\,\,\,\,\,\,\,\,\, = {\mathbb E_H}\left[ {F_{A\left| H \right.}^{(N)}\left( {\bar A\left| h \right.} \right)} \right] {\log _2}P + {\mathbb E_H}\left[ {F_{A\left| H \right.}^{(N)}\left( {\bar A\left| h \right.} \right)} {\log _2}\left( \frac{\theta}{\sigma^2} \right) + {{C_0}\left( {h,N,\bar A} \right)} \right],
   \]
   where ${\mathbb E_H}\left[ {F_{A\left| H \right.}^{(N)}\left( {\bar A\left| h \right.} \right)} {\log _2}\left( \frac{\theta}{\sigma^2} \right) + {{C_0}\left( {h,N,\bar A} \right)} \right]$ is a constant not related to $P$ and thus regarded as $o\left( {{{\log }_2}\,\eta} \right)$ since ${F_{A\left| H \right.}^{(N)}\left( {\bar A\left| h \right.} \right)} {\log _2}\left( \frac{\theta}{\sigma^2} \right) + {{C_0}\left( {h,N,\bar A} \right)}$ is a constant not related to $P$.

   For an integer $N \ge 1$, note that ${\Gamma \left( {N,x} \right)}$ is equivalently expressed as \cite{Proakis}
   \begin{equation}\label{App_IncompleteGamma}
      {{\Gamma \left( {N,x} \right)} = (N - 1)!\,{e^{ - x}}{\sum\limits_{m = 0}^{N-1} {\frac{x^m}{{m!}}} }.}
   \end{equation}
   From (\ref{Eq_TBS_CDF}), (\ref{Eq_pdf_H}), and (\ref{App_IncompleteGamma}), ${F_A^{(N)}}\left( a \right) =  {\mathbb E_H}\left[ {F_{A\left| H \right.}^{(N)}\left( {a\left| h \right.} \right)} \right]$ is obtained as
   \[ {F_A^{(N)}}\left( a \right) = \int_0^\infty  {{F_{A\left| H \right.}^{(N)}}\left( {a\left| h \right.} \right)} {f_H}\left( h \right)dh \,\,\,\,\,\,\,\,\,\,\,\,\,\,\,\,\,\,\,\,\,\,\,\,\,\,\,\,\,\,\,\,\,\,\,\,\,\,\,\,\,\,\,\,\,\,\,\,\,\,\,\,\,\,\,\,\,\,\,\,\,\,\,\,\,\,\,\,\,\,\,\,\,\,\,\,\,\,\,\,\,\,\,\,\,\,\, \]
   \[ \,\,\,\,\,\,\,\,\,\,\,\,\, = \int_0^\infty  {\left( {1 - {e^{ - \frac{{Na}}{h}}}\sum\limits_{m = 0}^{N - 1} {\frac{1}{{m!}}{{\left( {\frac{{Na}}{h}} \right)}^m}} } \right)\frac{{N_t^{{N_t}}}}{{\Gamma \left( {{N_t}} \right)}}{h^{{N_t} - 1}}{e^{ - {N_t}h}}dh} \]
   \begin{equation}\label{App_F_A_a1}
      {\,\,\,\,\,\, = 1 - \frac{{N_t^{{N_t}}}}{{\Gamma \left( {{N_t}} \right)}}\sum\limits_{m = 0}^{N - 1} {\frac{1}{{m!}}{{\left( {Na} \right)}^m}\underbrace {\int_0^\infty  {{h^{{N_t} - m - 1}}{e^{ - {N_t}h - \frac{{Na}}{h}}}dh} }_{\buildrel \Delta \over = \,\, \alpha} } .}
   \end{equation}
   By applying \cite[3.471-9]{TableOfIntegral} to $\alpha$, we can obtain (\ref{Eq_Marginal_CDF}). This completes the proof of Proposition \ref{Proposition_Avg_Scaling}.

   \section{Proof of Proposition \ref{Proposition_AvgEnergy}}\label{Proof_Proposition_AvgEnergy}
   From (\ref{Eq_Energy_TBS3}), (\ref{Eq_pdf_H}), and (\ref{App_IncompleteGamma}), ${\bar Q^{({\rm{T}})}}\left( {N,\bar A} \right) = {\mathbb E_H}\left[ {{Q^{\left( {\rm{T}} \right)}}\left( {h,N,\bar A} \right)} \right]$ is further obtained as
   \[
      {\bar Q^{({\rm{T}})}}\left( {N,\bar A} \right) = \theta P \mathbb E_H \left[ {{h}{e^{ - \frac{{N\bar A}}{{h}}}}\sum\limits_{k = 0}^N {\frac{1}{{k!}}{{\left( {N\bar A} \right)}^k}{{h}^{ - k}}} } \right] \,\,\,\,\,\,\,\,\,
   \]
   \[
      \,\,\,\,\,\,\,\,\,\,\,\,\,\,\,\,\,\,\,\,\,\,\,\, = \theta P \mathbb E_H \left[ {{e^{ - \frac{{N\bar A}}{{h}}}}\sum\limits_{k = 0}^N {\frac{1}{{k!}}{{\left( {N\bar A} \right)}^k}{{h}^{1 - k}}} } \right]
   \]
   \begin{equation}\label{App_AvgEnergy3}
      {\,\,\,\,\,\,\,\,\,\,\,\,\,\,\,\,\,\,\,\,\,\,\,\,\,\,\,\,\,\,\,\,\,\,\,\,\,\,\,\,\,\,\,\,\,\,\,\,\,\,\,\,\,\,\,\,\,\,\,\,\,\,\,\,\,\,\,\,\,\,\,\,\,\, = \theta P\frac{{N_t^{{N_t}}}}{{\Gamma \left( {{N_t}} \right)}}\sum\limits_{k = 0}^N {\frac{1}{{k!}}{{\left( {N\bar A} \right)}^k}\underbrace {\int_0^\infty  {{h^{{N_t} - k}}{e^{ - {N_t}h - \frac{{N\bar A}}{h}}}dh} }_{\buildrel \Delta \over = \,\, \beta} } . \,\,\,\,\,\,\,\,\,\,\,\,\,}
   \end{equation}
   By applying \cite[3.471-9]{TableOfIntegral} to $\beta$ in (\ref{App_AvgEnergy3}), we obtain (\ref{Eq_AvgEnergy_Slope}). This completes the proof of Proposition \ref{Proposition_AvgEnergy}.

   \section{Proof of Theorem \ref{Theorem_Optimality_N}}\label{App_Proof_Theorem_Optimality_N}
   First, the former part of Theorem \ref{Theorem_Optimality_N} can be proved using the following lemma by considering an arbitrary distribution of $A$ with PDF and CDF denoted by ${g_A}\left( a \right)$ and ${G_A}\left( a \right)$, respectively, where ${G_A}\left( a \right) > 0$ for $a > 0$.
   \begin{lemma}\label{Lemma_OOS_vs_TBS}
      Given ${G_A}\left( {\bar A} \right) = \tau$ with $0 < \tau < 1$ and $0 < \bar A < \infty$,
      \begin{equation}\label{Eq_Lemma_OOS_vs_TBS}
         {\int_{\bar A}^\infty  {a{g_A}\left( a \right)da} > \left( {1 - \tau } \right)b.}
      \end{equation}
   \end{lemma}
   \begin{proof}
      Please refer to Appendix \ref{App_Proof_Lemma_TBS_vs_OOS}.
   \end{proof}

   For TS, we have the rate scaling factor ${{\Delta^{\left( {\rm{T}} \right)}}\left( {N,\bar A} \right) = {F_A^{(N)}}\left( \bar A \right)}$ from Proposition \ref{Proposition_Avg_Scaling}, and it can be shown from (\ref{Eq_Marginal_CDF}) that ${F_A^{(N)}}\left( a \right) > 0$ for $a > 0$. In addition, the energy scaling factor for TS can be alternatively expressed as ${{\Pi^{\left( {\rm{T}} \right)}}\left( {N,\bar A} \right) = \int_{\bar A}^{\infty} {a{{f_A^{(N)}}\left( a \right)}da}}$ with ${f_A^{(N)}}\left( a \right) = \mathbb E_H \left[{f_{A\left| H \right.}^{(N)}}\left( {a\left| h \right.} \right)\right]$ denoting the unconditional PDF of $A$ after averaging over the fading distribution. Furthermore, it can be easily verified that $\int_{0}^{\infty} {a{{f_A^{(N)}}\left( a \right)}da} = 1$. Given ${F_A^{(N)}}\left( \bar A \right) = \tau$ with $0 < \bar A < \infty$ and $0 < \tau < 1$, it can thus be verified from Lemma \ref{Lemma_OOS_vs_TBS} that ${{\Delta^{\left( {\rm{T}} \right)}}\left( {N,\bar A} \right)} + {\Pi^{\left( {\rm{T}} \right)}}\left( {N,\bar A} \right) > 1$, by substituting $b$, ${{g_A}\left( a \right)}$, and ${{G_A}\left( a \right)}$ in Lemma \ref{Lemma_OOS_vs_TBS} by $1$, ${f_A^{(N)}}\left( a \right)$, and ${F_A^{(N)}}\left( a \right)$, respectively. Since we have ${\Delta^{({\rm{P}})}}\left( \tau  \right) + {\Pi^{({\rm{P}})}}\left( \tau  \right) = 1$ for PS from (\ref{Eq_PreLog_OOS}) and (\ref{AvgEnergy_OOS}), it follows that ${{\Delta^{\left( {\rm{T}} \right)}}\left( {N,\bar A} \right)} + {\Pi^{\left( {\rm{T}} \right)}}\left( {N,\bar A} \right) > {\Delta^{({\rm{P}})}}\left( \tau  \right) + {\Pi^{({\rm{P}})}}\left( \tau  \right)$. Therefore, we have ${{\Delta^{\left( {\rm{T}} \right)}}\left( {N,\bar A} \right)} > {\Delta^{({\rm{P}})}}\left( \tau  \right)$ for given $0 < {\Pi^{\left( {\rm{T}} \right)}}\left( {N,\bar A} \right) = {\Pi^{({\rm{P}})}}\left( \tau  \right) < 1$. This proves the former part of Theorem \ref{Theorem_Optimality_N}.

   It is worth remarking that Lemma \ref{Lemma_OOS_vs_TBS} implies that TS in general yields better trade-off between the rate and energy scaling factors than PS provided that the average received channel power for TS is the same as that for PS, based on which the former part of Theorem \ref{Theorem_Optimality_N} for the i.i.d. Rayleigh fading MISO channel with a fixed threshold $\bar A$ is proved. As another example, even for a transmission block with $H = h$, TS with $N$ RBs, $1 \le N \le N_t$, yields better trade-off between the rate and energy scaling factors than PS. This can be proved by substituting $b$, ${{g_A}\left( a \right)}$, and ${{G_A}\left( a \right)}$ in Lemma \ref{Lemma_OOS_vs_TBS} by $h$, $f_{A\left| H \right.}^{(N)}( {\bar A\left| h \right.} )$ in (\ref{Eq_TBS_pdf}), and $F_{A\left| H \right.}^{(N)}( {\bar A\left| h \right.} )$ in (\ref{Eq_TBS_CDF}), respectively. This is originated from the fact that for both TS and PS schemes the rate scaling factor is determined by the percentage of sub-blocks allocated to ID mode, whereas the energy scaling factor is determined by the percentage of sub-blocks assigned to EH mode as well as their channel power values. Note that TS scheme assigns a subset of sub-blocks with the largest channel power to EH mode, as inferred from (\ref{Eq_Opt_1_Solution}). Therefore, given a percentage of sub-blocks allocated to EH mode, i.e., $1 - {G_{A}}\left( {\bar A} \right) = 1 - \tau$ for TS and PS schemes, respectively, the energy scaling factor of TS is larger than that of PS while rate scaling factors are the same for both schemes, i.e., ${G_{A}}\left( {\bar A} \right) = \tau$.

   Next, to prove the latter part of Theorem \ref{Theorem_Optimality_N}, we consider two arbitrary distributions of $A$ with PDFs denoted by ${g_A}\left( a \right)$ and ${u_A}\left( a \right)$, and the corresponding CDFs denoted by ${G_A}\left( a \right)$ and ${U_A}\left( a \right)$, respectively. It is assumed that $\int_{0}^\infty  {a{g_A}\left( a \right)da} = \int_{0}^\infty  {a{u_A}\left( a \right)da} = b > 0$. It is further assumed that ${G_A}\left( a \right) > 0$ and ${U_A}\left( a \right) > 0$ for $a > 0$, and ${G_A}\left( a \right)$ and ${U_A}\left( a \right)$ intersect at $a = \hat A$, satisfying
   \begin{equation}\label{Eq_Condition2}
      {\left\{ {\begin{array}{*{20}{c}}
      {{G_{A}}\left( {a} \right) > {U_{A}}\left( {a} \right),\,\,\,{\rm{if}}\,\,0 < a < \hat A}  \\
      {{G_{A}}\left( {a} \right) = {U_{A}}\left( {a} \right),\,\,\,{\rm{if}}\,\,a = \hat A\,\,\,\,\,\,\,\,\,\,\,}  \\
      {{G_{A}}\left( {a} \right) < {U_{A}}\left( {a} \right),\,\,\,{\rm{if}}\,\, a > \hat A. \,\,\,\,\,\,\,\,\,}  \\
      \end{array}} \right.}
   \end{equation}

   \begin{lemma}\label{Lemma_Optimality}
      Given $0 < G_A\left( {\bar A_g} \right) = U_A\left( {\bar A_u} \right) < 1$ with $0 < \bar A_g, \bar A_u < \infty$,
      \begin{equation}\label{Eq_Lemma_TBS_Optimality}
         {\int_{\bar A_g}^\infty  {a{g_A}\left( a \right)da} > \int_{\bar A_u}^\infty  {a{u_A}\left( a \right)da}.}
      \end{equation}
   \end{lemma}
   \begin{proof}
      Please refer to Appendix \ref{App_Proof_Lemma_Optimality}.
   \end{proof}

   The latter part of Theorem \ref{Theorem_Optimality_N} can be proved using Lemma \ref{Lemma_Optimality} as follows. Given $1 \le N < M \le N_t$ for TS, it can be verified that $\int_0^\infty  {af_A^{(N)}\left( a \right)da}  = \int_0^\infty  {af_A^{(M)}\left( a \right)da}  = 1$. Furthermore, it can be shown from (\ref{Eq_Marginal_CDF}) that ${{F_A^{(N)}}\left( a \right)}$ and ${{F_A^{(M)}}\left( a \right)}$ correspond to ${G_A}\left( a \right)$ and ${U_A}\left( a \right)$ in (\ref{Eq_Condition2}), respectively (c.f. Fig. \ref{Fig_Marginal_CDF}). By substituting ${{f_A^{(N)}}\left( a \right)}$, ${{f_A^{(M)}}\left( a \right)}$, ${{F_A^{(N)}}\left( a \right)}$, and ${{F_A^{(M)}}\left( a \right)}$ for ${{g_A}\left( a \right)}$, ${{u_A}\left( a \right)}$, ${{G_A}\left( a \right)}$, and ${{U_A}\left( a \right)}$ in Lemma \ref{Lemma_OOS_vs_TBS}, respectively, it can be verified that ${\Pi^{\left( {\rm{T}} \right)}}\left( {N,\bar A_N} \right) > {\Pi^{\left( {\rm{T}} \right)}}\left( {M,\bar A_M} \right)$ when ${{\Delta^{\left( {\rm{T}} \right)}}\left( {N,\bar A_N} \right)} = {{\Delta^{\left( {\rm{T}} \right)}}\left( {M,\bar A_M} \right)}$, $0 < \bar A_N, \bar A_M < \infty$ since ${{\Delta^{\left( {\rm{T}} \right)}}\left( {N,\bar A_N} \right)} = {{F_A^{(N)}}\left( a \right)}$ from Proposition \ref{Proposition_Avg_Scaling} and ${{\Pi^{\left( {\rm{T}} \right)}}\left( {N,\bar A_N} \right) = \int_{\bar A_N}^{\infty} {a{{f_A^{(N)}}\left( a \right)}da}}$. This guarantees that ${{\Delta^{\left( {\rm{T}} \right)}}\left( {N,\bar A_N} \right)} > {{\Delta^{\left( {\rm{T}} \right)}}\left( {M,\bar A_M} \right)}$ for given $0 < {\Pi^{\left( {\rm{T}} \right)}}\left( {N,\bar A_N} \right) = {\Pi^{\left( {\rm{T}} \right)}}\left( {M,\bar A_M} \right) < 1$, since both ${{\Delta^{\left( {\rm{T}} \right)}}\left( {N,\bar A_N} \right)}$ and ${{\Delta^{\left( {\rm{T}} \right)}}\left( {M,\bar A_M} \right)}$ decrease monotonically with increasing ${\Pi^{\left( {\rm{T}} \right)}}\left( {N,\bar A_N} \right)$ and ${\Pi^{\left( {\rm{T}} \right)}}\left( {M,\bar A_M} \right)$, respectively. This proves the latter part of Theorem \ref{Theorem_Optimality_N}.

   As a remark, Lemma \ref{Lemma_Optimality} compares the trade-offs between the rate and energy scaling factors in TS schemes with two different distributions of channel power induced by different values of $N$ provided that both distributions have the same average channel power and satisfy the condition in (\ref{Eq_Condition2}). The latter part of Theorem \ref{Theorem_Optimality_N} for the i.i.d Rayleigh fading MISO channel with a fixed $\bar A$ is one application of Lemma \ref{Lemma_Optimality}. As another example, even for a transmission block with $H = h$, better trade-off between the rate and energy scaling factors is attained with $N$ than $M$ RBs, $1 \le N < M \le N_t$, since $F_{A\left| H \right.}^{(N)}\left( {a\left| h \right.} \right)$ and $F_{A\left| H \right.}^{(M)}\left( {a\left| h \right.} \right)$ correspond to ${G_A}\left( a \right)$ and ${U_A}\left( a \right)$ in (\ref{Eq_Condition2}), respectively, as shown from (\ref{Eq_TBS_CDF}). This is due to the fact that the artificial channel fading is more substantial when smaller number of RBs is employed, and the argument similarly as for Lemma \ref{Lemma_OOS_vs_TBS}.

   Combining the proofs for both the above two parts, Theorem \ref{Theorem_Optimality_N} is proved.

   \section{Proof of Lemma \ref{Lemma_OOS_vs_TBS}}\label{App_Proof_Lemma_TBS_vs_OOS}
   Integrating by part, $\int_0^{\bar A} {a{g_A}\left( a \right)da}$ can be evaluated as
   \begin{equation}\label{App_Integration_by_Part}
      {\int_0^{\bar A} {a{g_A}\left( a \right)da}  = \bar A{G_A}\left( {\bar A} \right) - \int_0^{\bar A} {{G_A}\left( a \right)da}.}
   \end{equation}

   Assume that $\bar A$ is given such that ${G_{A}}\left( {\bar A} \right) = \tau$, $0 < \tau < 1$. From (\ref{Eq_Lemma_OOS_vs_TBS}) and (\ref{App_Integration_by_Part}), we have
   \[ \tilde E = \int_{\bar A}^\infty  {a{g_A}\left( a \right)da}  - \left( {1 - \tau } \right)b \,\,\,\,\,\,\,\,\,\,\,\,\,\,\,\,\,\,\,\,\,\,\,\,\,\,\,\,\,\,\,\,\,\,\,\,\,\, \]
   \[ \,\,\,\,\,\,\,\,\,\,\,\,\,\,\,\,\,\,\,\,\,\,\,\,\,\,\,\,\, = \int_0^\infty  {a{g_A}\left( a \right)da}  - \int_0^{\bar A} {a{g_A}\left( a \right)da}  - b\left( {1 - {G_A}\left( {\bar A} \right)} \right)\]
   \begin{equation}\label{App_EnergyGap_OOS}
      { = \left( {b - \bar A} \right){G_A}\left( {\bar A} \right) + \int_0^{\bar A} {{G_A}\left( {\bar A} \right)da,} \,\,\,\,\,\,\,\,\,\,\,\,\, }
   \end{equation}
   and the resulting derivative
   \begin{equation}\label{App_Differentiation_OOS}
      {\mu = \frac{{d\tilde E}}{{d\bar A}} = \left( {b - \bar A} \right){g_A}\left( \bar A \right),}
   \end{equation}
   where $\mu = 0$ is achieved at $\bar A = b$.

   From (\ref{App_EnergyGap_OOS}) and (\ref{App_Differentiation_OOS}), it is observed that $\tilde E > 0$ with $0 < \bar A \le b$, and $\tilde E$ increases monotonically with $\bar A$ until $\bar A = b$. For $\bar A > b$, it is observed that $\mu < 0$ and thus $\tilde E$ decreases monotonically with increasing $\bar A$. Because $\mathop {\lim }\limits_{\bar A \to \infty } \int_{\bar A}^\infty  {a{g_A}\left( a \right)da}  = 0$, it can be shown that  $\tilde E > 0$ even with $\bar A > b$. Therefore, it is verified that $\tilde E > 0$ in (\ref{App_EnergyGap_OOS}) for any $\bar A > 0$. This completes the proof of Lemma \ref{Lemma_OOS_vs_TBS}.

\clearpage

   \section{Proof of Lemma \ref{Lemma_Optimality}}\label{App_Proof_Lemma_Optimality}
   Denote $\Delta = G_A( {\bar A_g} ) = U_A( {\bar A_u} )$, $0 < \bar A_g, \bar A_u < \infty$. From (\ref{Eq_Lemma_TBS_Optimality}) and (\ref{App_Integration_by_Part}), we have
   \[\tilde E = \int_{{{\bar A}_g}}^\infty  {a{g_A}\left( a \right)da}  - \int_{{{\bar A}_u}}^\infty  {a{u_A}\left( a \right)da} \,\,\,\,\,\,\,\,\,\,\,\,\,\,\,\,\,\,\,\,\,\, \]
   \begin{equation}\label{App_EnergyGap_Temp}
      { = \int_0^{{{\bar A}_u}} {a{u_A}\left( a \right)da}  - \int_0^{\bar A_g} {a{g_A}\left( a \right)da} \,\,\,\,\,\,\,\,\,\,\,\,\,\, }
   \end{equation}
   \begin{equation}\label{App_EnergyGap_General}
      { \,\,\,\,\,\,\,\,\,\,\,\,\,\,\,\,\,\,\,\,\,\,\,\, = \Delta \left( {{{\bar A}_u} - {{\bar A}_g}} \right) + \int_0^{{{\bar A}_g}} {{G_{A}}\left( {a} \right)da}  - \int_0^{{{\bar A}_u}} {{U_{A}}\left( {a} \right)da}.}
   \end{equation}

   According to (\ref{Eq_Condition2}), there are three cases addressed as follows for given $0 < \Delta < 1$.

      \emph{1)} ${{\bar A}_g} = {{\bar A}_u} = \hat A$: According to (\ref{Eq_Condition2}), $G_A ( {\hat A} ) = U_A ( {\hat A}  ) = \hat \Delta$. Since it is assumed in (\ref{Eq_Condition2}) that ${G_{A}}( {a} ) > {U_{A}}( {a} )$ with $0 < a < \hat A$, $\tilde E$ is evaluated from (\ref{App_EnergyGap_General}) as
      \begin{equation}\label{App_EnergyGap_Equal}
         {\tilde E = \int_0^{\hat A} {\left( {{G_{A}}\left( {a} \right) - {U_{A}}\left( {a} \right)} \right)da}  > 0.}
      \end{equation}

      \emph{2)} $0 < {{\bar A}_g}, {{\bar A}_u} < \hat A$: It can be inferred from (\ref{Eq_Condition2}) that ${{\bar A}_g} < {{\bar A}_u} < \hat A$, which results in $0 < \Delta < \hat \Delta$. From (\ref{App_EnergyGap_General}), we have
      \begin{equation}\label{App_EnergyGap_Smaller}
         {\tilde E = \underbrace {U_A\left( \bar A_u \right) \left( {{{\bar A}_u} - {{\bar A}_g}} \right) - \int_{{{\bar A}_g}}^{{{\bar A}_u}} {{U_{A}}\left( {a} \right)da} }_{\buildrel \Delta \over = \,\, \beta}  + \underbrace {\int_0^{{{\bar A}_g}} {\left( {{G_{A}}\left( {a} \right) - {U_{A}}\left( {a} \right)} \right)da} }_{\buildrel \Delta \over = \,\, \alpha} .}
      \end{equation}
      Since ${{\bar A}_g} < {{\bar A}_u}$ and ${G_{A}}( {a} ) > {U_{A}}( {a} )$ with $0 < a < \hat A$, it can be verified that $\alpha > 0$ and $\beta > 0$, and thus $\tilde E > 0$.

      \emph{3)} $\hat A < {{\bar A}_g}, {{\bar A}_u} < \infty$: It can be inferred from (\ref{Eq_Condition2}) that $\hat A < {{\bar A}_u} < {{\bar A}_g}$, which results in $\hat \Delta < \Delta < 1$. From (\ref{App_EnergyGap_Temp}), we have
      \begin{equation}\label{App_EnergyGap_Larger}
         {\tilde E = \underbrace {\int_0^{\hat A} {a\left( {{u_{A}}\left( {a} \right) - {g_{A}}\left( {a} \right)} \right)da} }_{\buildrel \Delta \over = \,\, \delta}  - \underbrace {\left( {\int_{\hat A}^{{{\bar A}_g}} {a{g_{A}}\left( {a} \right)da}  - \int_{\hat A}^{{{\bar A}_u}} {a{u_{A}}\left( {a} \right)da} } \right)}_{\buildrel \Delta \over = \,\, \varepsilon} ,}
      \end{equation}
      with $\delta > 0$ as shown in (\ref{App_EnergyGap_Temp})-(\ref{App_EnergyGap_Equal}). In addition, it can be verified that
      \begin{equation}\label{App_Limitation1}
         {\mathop {\lim }\limits_{\Delta  \to \hat \Delta } \varepsilon  = \mathop {\lim }\limits_{{{\bar A}_g} \to \hat A} \int_{\hat A}^{{{\bar A}_g}} {a{g_{A}}\left( {a} \right)da}  - \mathop {\lim }\limits_{{{\bar A}_u} \to \hat A} \int_{\hat A}^{{{\bar A}_u}} {a{u_{A}}\left( {a} \right)da}  = 0,}
      \end{equation}
      \begin{equation}\label{App_Limitation1}
         {\mathop {\lim }\limits_{\Delta  \to 1} \varepsilon  = \mathop {\lim }\limits_{\scriptstyle {{\bar A}_g} \to \infty , \hfill \atop
         \scriptstyle {{\bar A}_u} \to \infty  \hfill} \varepsilon  = \int_0^{\hat A} {a{u_{A}}\left( {a} \right)da}  - \int_0^{\hat A} {a{g_{A}}\left( {a} \right)da}  = \delta .}
      \end{equation}
      Since $\frac{{d\Delta }}{{d{{\bar A}_g}}} = \frac{d}{{d{{\bar A}_g}}}{G_{A}}\left( {\bar A_g} \right) = {g_{A}}\left( {\bar A_g} \right)$ and $\frac{{d\Delta }}{{d{{\bar A}_u}}} = \frac{d}{{d{{\bar A}_u}}}{U_{A}}\left( {\bar A_u} \right) = {u_{A}}\left( {\bar A_u} \right)$, we have
      \[ \frac{d}{{d\Delta }}\varepsilon  = \frac{1}{{\frac{{d\Delta }}{{d{{\bar A}_g}}}}}\frac{d}{{d{{\bar A}_g}}}\int_{\hat A}^{{{\bar A}_g}} {a{g_{A}}\left( {a} \right)da}  - \frac{1}{{\frac{{d\Delta }}{{d{{\bar A}_u}}}}}\frac{d}{{d{{\bar A}_u}}}\int_{\hat A}^{{{\bar A}_u}} {a{u_{A}}\left( {a} \right)da} \]
      \begin{equation}\label{App_Differentiation}
         {\,\,\,\, \,\,\,\, = {{\bar A}_g} - {{\bar A}_u} > 0. \,\,\,\,\,\,\,\,\,\,\,\,\,\,\,\,\,\,\,\,\,\,\,\,\,\,\,\,\,\,\,\,\,\,\,\,\,\,\,\,\,\,\,\,\,\,\,\,\,\,\,\,\,\,\,\,\,\,\,\,\,\,\,\,\,\,\,\,\,\,\,\,\,\,\,\,\,\,\,\,\,\,\,\,\,\,\,\,\,\,\,\,}
      \end{equation}
      From (\ref{App_EnergyGap_Larger})-(\ref{App_Differentiation}), it can be verified that ${\tilde E}$ monotonically decreases from $\delta$ to $0$ with increasing $\Delta$ with $\hat \Delta < \Delta < 1$, i.e., ${\tilde E} > 0$ with $\hat A < {{\bar A}_g}, {{\bar A}_u} < \infty$.

   Combining the above three cases, Lemma \ref{Lemma_Optimality} is thus proved.

   \section{Proof of Proposition \ref{Proposition_EnergyDiversity_TBS}}\label{App_Proof_Proposition_EnergyDiversity_TBS}
   Given a transmission block with $H = h$, ${Q^{(\rm{T})}}\left( {h,N,0} \right) = \theta Ph$ from (\ref{Eq_Energy_TBS3}). In the i.i.d. Rayleigh fading MISO channel, given $\hat Q > 0$, $p_{Q,out}^{({\rm{T}})}$ for TS with $\bar A = 0$ can be approximated as $P \to \infty$ by $\mathop {\lim }\limits_{P \to \infty } \Pr\left( {h < \frac{\hat Q}{\theta P}} \right) = \mathop {\lim }\limits_{P \to \infty } F_{H}\left( {\frac{\hat Q}{\theta P}} \right) = {\left( {\frac{\hat Q}{{\theta P }}} \right)^{ N_t}}$, since ${\mathop {\lim }\limits_{h \to 0 } F_{H}\left( {h} \right) = {h ^{N_t}}}$ in (\ref{Eq_CDF_H}). This proves the first equality in (\ref{Eq_EnergyDiversity_TBS}) for $\bar A = 0$.

   When $\bar A > 0$, the harvested power per block for a given $h$, ${Q^{(\rm{T})}}\left( {h,N,\bar A} \right)$ in (\ref{Eq_Energy_TBS3}), is a monotonically increasing function of $h$, since ${\Gamma \left( {\alpha,x} \right)}$ is a monotonically decreasing function of $x$. For a given power requirement $\hat Q > 0$, denote $\bar h$ as the minimum value of $h$ such that ${Q^{(\rm{T})}}\left( {\bar h,N,\bar A} \right) \ge \hat Q$, i.e.,
   \begin{equation}\label{App_g_hbar}
      {\vartheta\left( {\bar h} \right) \buildrel \Delta \over = \bar h\,\frac{{\Gamma \left( {N + 1,N\bar A/\bar h} \right)}}{{\Gamma \left( {N + 1} \right)}} = \frac{\hat Q}{\theta P}.}
   \end{equation}
   Thus, the power outage probability for TS with given $N$, $\bar A$, and $\hat Q$ is obtained as $p_{Q, \,out}^{({\rm{T}})}( {N, \bar A,\hat Q} ) = F_H(\bar h)$. Since $\vartheta\left( {h} \right)$ increases with $h$, with $P \to \infty$ it then follows from (\ref{App_g_hbar}) that $\bar h \to 0$, under which we have
   \begin{equation}\label{App_q_h_temp}
      {\vartheta\left( \bar h \right) = \bar h {e^{ - \left( {N\bar A/\bar h} \right)}}{\sum\limits_{k = 0}^N {\frac{1}{{k!}}\left( {\frac{{N\bar A}}{\bar h}} \right)} ^k}}
   \end{equation}
   \begin{equation}\label{App_q_h}
      {\,\,\,\,\,\,\,\,\,\,\,\, \approx \bar h {e^{ - \left( {N\bar A/\bar h} \right)}}\frac{1}{{N!}}{\left( {\frac{{N\bar A}}{\bar h}} \right)^N},}
   \end{equation}
   where, since $N \ge 1$ is an integer, (\ref{App_q_h_temp}) is obtained from \cite{Proakis}
   \[
      {{\Gamma \left( {N,x} \right)} = (N - 1)!\,{e^{ - x}}{\sum\limits_{m = 0}^{N-1} {\frac{x^m}{{m!}}} }.}
   \]
   Since $\ln \vartheta\left( {\bar h} \right) = \ln \frac{\hat Q}{\theta P}$, from (\ref{App_q_h}), we have
   \begin{equation}\label{App_EnergyOutage1}
      {{\left( {N - 1} \right)\ln x - N\bar Ax = \ln \left( {\frac{{\hat QN!}}{{\theta P{{\left( {N\bar A} \right)}^N}}}} \right)},}
   \end{equation}
   where $x = 1/\bar h$. With $\bar h \to 0$, i.e., $x \to \infty$, the left-hand side of (\ref{App_EnergyOutage1}) can be further approximated as $- N\bar Ax$. Therefore, $\bar h$ as $P \to \infty$ can be approximated by
   \begin{equation}\label{App_EnergyOutage2}
      {\bar h = {N\bar A{{\left( {\ln \left( {\frac{{\theta P{{\left( {N\bar A} \right)}^N}}}{{\hat QN!}}} \right)} \right)}^{ - 1}}} \approx {N\bar A {{\left( {\ln \left( \theta P \right)} \right)}^{ - 1}}}.}
   \end{equation}
   From (\ref{App_EnergyOutage2}) and the fact that ${F_H}\left( h \right) \approx {h^{ N_t}}$ as $h \to 0$, we obtain the second equality in (\ref{Eq_EnergyDiversity_TBS}) for $\bar A > 0$.

   From the proofs for both the first and second equalities in (\ref{Eq_EnergyDiversity_TBS}), Proposition \ref{Proposition_EnergyDiversity_TBS} is thus proved.

   \section{Derivation of (\ref{Eq_BRB_EnergyOutage})}\label{App_Derivation_P_E_out_B}
   From (\ref{Eq_AS_Energy}), the energy outage probability of TS-B is obtained as
   \begin{equation}\label{App_Outage_Total}
      {p_{Q,out}^{({\rm{B}})} = \Pr \left( {v < \bar A} \right) + \Pr \left( {w \le \bar A \le v, \, \frac{{\theta Pv}}{2} < \hat Q} \right) + \Pr \left( {w > \bar A, \, \frac{{\theta P\left( {v + w} \right)}}{2} < \hat Q} \right).}
   \end{equation}

   First, $\Pr \left( {v < \bar A} \right)$ in (\ref{App_Outage_Total}) is given by
   \begin{equation}\label{App_Outage_Case1}
      {\Pr \left( {v < \bar A} \right) = {F_V}\left( \bar A \right) = {\left( {1 - {e^{ - \bar A}}} \right)^2},}
   \end{equation}
   where ${F_V}\left( v \right)$ denotes the CDF of $V = \max ( \, {{{\left| {{h_1}} \right|}^2},{{\left| {{h_2}} \right|}^2}} )$, given by ${F_V}\left( v \right) = {\left( {1 - {e^{ - v}}} \right)^2}$, since both ${{\left| {{h_1}} \right|}^2}$ and ${{\left| {{h_2}} \right|}^2}$ are independent exponential random variables.

   The second term in (\ref{App_Outage_Total}) can be obtained as
   \[\Pr \left( {w \le \bar A \le v, \, \frac{{\theta Pv}}{2} < \hat Q} \right) = \Pr \left( {w \le \bar A, \, \bar A \le v \le 2D, \, w \le v} \right) \,\,\,\,\,\,\,\,\,\,\,\,\,\,\,\,\,\,\,\,\,\,\,\,\,\,\,\,\,\,\,\,\,\,\,\,\,\,\,\,\,\,\,\,\,\,\,\,\, \]
   \[ \,\,\,\,\,\,\,\,\,\,\,\,\,\,\,\,\,\,\,\,\,\,\,\,\,\,\,\,\,\,\,\,\,\,\,\,\,\,\,\, = {\bf{1}}\left( {\bar A < 2D} \right) \int_{\bar A}^{2D} {\int_0^{\bar A} {{f_{V,W}}\left( {v,w} \right)dwdv} } \]
   \begin{equation}\label{App_Outage_Case2}
      { \,\,\,\,\,\,\,\,\,\,\,\,\,\,\,\,\,\,\,\,\,\,\,\,\,\,\,\,\,\,\,\,\,\,\,\,\,\,\,\,\,\,\,\,\,\,\,\,\,\,\,\,\,\,\,\,\,\,\,\,\,\,\,\,\,\,\,\, = {\bf{1}}\left( {\bar A < 2D} \right) 2{e^{ - 2\left( {\bar A + D} \right)}}\left( { - 1 + {e^{\bar A}}} \right)\left( { - {e^{\bar A}} + {e^{2B}}} \right),}
   \end{equation}
   where $D = \frac{\hat Q}{\theta P}$ and ${f_{V,W}}\left( {v,w} \right)$ denotes the joint PDF for $V$ and $W$ given by ${f_{V,W}}\left( {v,w} \right) = 2{e^{ - v}}{e^{ - w}}$, $v > w$.

   Finally, the last term in (\ref{App_Outage_Total}) can be obtained as
   \[ \Pr \left( {w > \bar A, \, \frac{{\theta P\left( {v + w} \right)}}{2} < \hat E} \right)  \,\,\,\,\,\,\,\,\,\,\,\,\,\,\,\,\,\,\,\,\,\,\,\,\,\,\,\,\,\,\,\,\,\,\,\,\,\,\,\,\,\,\,\,\,\,\,\,\,\,\,\,\,\,\,\,\,\,\, \]
   \[ = \Pr \left( {w > \bar A, \, v + w < 2D } \right) \,\,\,\,\,\,\,\,\,\,\,\,\,\,\,\,\,\,\,\,\,\,\,\,\,\,\,\,\,\,\,\,\,\,\,\,\,\,\,\,\,\,\,\,\,\,\,\,\,\,\,\,\,\,\,\,\,\,\,\,\,\,\,\,\,\,\,\,\,\,\,\,\,\,\,\,\,\,\, \]
   \[ \,\,\,\,\,\,\,\,\,\,\,\,\,\,\,\,\,\,\,\,\,\,\,\,\,\,\,\,\,\,\,\,\,\,\,\,\,\,\,\,\,\,\,\,\,\,\,\,\,\, = {\bf{1}}\left( {\bar A < D} \right) \left( \int_{\bar A}^B {\int_{\bar A}^v {{f_{V,W}}\left( {v,w} \right)dwdv} }  + \int_B^{2B - \bar A} {\int_{\bar A}^{2B - v} {{f_{V,W}}\left( {v,w} \right)dwdv} }  \right)  \]
   \begin{equation}\label{App_Outage_Case3}
      { \,\,\,\,\,\,\,\,\,\,\,\,\,\,\,\,\,\,\,\,\,\,\,\,\,\,\,\,\,\,\,\,\,\,\,\,\,\,\,\,\,\,\,\,\,\,\,\,\,\,\,\,\,\,\,\,\,\,\,\, = {\bf{1}}\left( {\bar A < D} \right)\left( {{e^{ - 2\left( {\bar A + D} \right)}}{{\left( {{e^{\bar A}} - {e^D}} \right)}^2} + {e^{ - \bar A - 2D}}\left( {\left( { - 1 + \bar A - D} \right){e^{\bar A}} + {e^D}} \right)} \right).}
   \end{equation}

   From (\ref{App_Outage_Total})-(\ref{App_Outage_Case3}), we can obtain (\ref{Eq_BRB_EnergyOutage}).


\begin{thebibliography}{1}
\bibliographystyle{IEEEbib}

   %% Simultaneous WET & WIT
   \bibitem{Zhou}
      X. Zhou, R. Zhang, and C. K. Ho, ``Wireless information and power transfer: architecture design and rate-energy tradeoff,'' \emph{IEEE Trans. Commun.}, vol. 61, no. 11, pp. 4757-4767, Nov. 2013.

   \bibitem{Liu}
      L. Liu, R. Zhang, and K. C. Chua, ``Wireless information transfer with opportunistic energy harvesting,''  \emph{IEEE Trans. Wireless Commun.}, vol. 12, no. 1, pp. 288-300, Jan. 2013.

   \bibitem{Liang}
      L. Liu, R. Zhang, and K. C. Chua, ``Wireless information and power transfer: a dynamic power splitting approach,'' \emph{IEEE Trans. Commun.}, vol. 61, no. 9, pp. 3990-4001, Sep. 2013.

   \bibitem{Caspers}
      E. P. Caspers, S. H. Yeung, T. K. Sarkar, A. Garcia-Lamperez, M. S. Palma, M. A. Lagunas, and A. Perez-Neira, ``Analysis of information and power transfer in wireless communications,'' \emph{IEEE Antennas and Propagation Magazine}, vol. 55, no. 3, pp.82-95, June 2013.

   \bibitem{Zhang}
      R. Zhang and C. K. Ho, ``MIMO broadcasting for simultaneous wireless information and power transfer,'' \emph{IEEE Trans. Wireless Commun.}, vol. 12, no. 5, pp. 1989-2001, May, 2013.

   \bibitem{Xiang}
      Z. Xiang and M. Tao, ``Robust beamforming for wireless information and power transmission,'' \emph{IEEE Wireless Commun. Letters}, vol. 1, no. 4, pp. 372-375, Aug. 2012.

   \bibitem{Chalise}
      B. K. Chalise, Y. D. Zhang, and M. G. Amin, ``Energy harvesting in an OSTBC based amlify-and-forward MIMO relay system,'' in \emph{Proc. IEEE Int. Conf. on Acousitcs, Speech, and Signal Process. (ICASSP)}, pp. 3201-3204, Mar. 2012.

   \bibitem{Xu}
      J. Xu, L. Liu, and R. Zhang, ``Multiuser MISO beamforming for simultaneous wireless information and power transfer,'' in \emph{Proc. IEEE Int. Conf. on Acousitcs, Speech, and Signal Process. (ICASSP)}, May 2013.

   \bibitem{Timotheou}
      S. Timotheou, I. Krikidis, and B. Ottersten, ``MISO interference channel with QoS and RF energy harvesting constraints,'' in \emph{Proc. IEEE Int. Conf. Commun. (ICC)}, June 2013.

   \bibitem{Park}
      J. Park and B. Clerckx, ``Joint wireless information and energy transfer in a two-user MIMO interference channel,'' \emph{IEEE Trans. Wireless Commun.}, vol. 12, no. 8, pp. 4210-4221, Aug. 2013.

   \bibitem{Gurakan}
      B. Gurakan, O. Ozel, J. Yang, and S. Ulukus, ``Energy cooperation in energy harvesting wireless systems,'' in \emph{Proc. IEEE Int. Symp. Inf. Theory (ISIT)}, pp. 965-969, July 2012.

   \bibitem{Narir}
      A. A. Nasir, X. Zhou, S. Durrani, and R. A. Kennedy, ``Relaying protocols for wireless energy harvesting and information processing,'' \emph{IEEE Trans. Wireless Commun.}, vol. 12, no. 7, pp. 3622-3636, July 2013.

   \bibitem{Ng}
      D. W. K. Ng, E. S. Lo, and R. Schober, ``Energy-efficient resource allocation in multiuser OFDM systems with wireless information and power transfer,'' in \emph{Proc. IEEE Wireless Commun. and Networking Conf. (WCNC)}, Apr. 2013.

   \bibitem{Huang}
      K. Huang and E. G. Larsson, ``Simultaneous information and power transfer for braodband wireless systems,'' in \emph{Proc. IEEE Int. Conf. on Acousitcs, Speech, and Signal Process. (ICASSP)}, May 2013.

   \bibitem{Zhou2}
      X. Zhou, R. Zhang, and C. K. Ho, ``Wireless information and power transfer in multiuser OFDM systems,'' submitted for publication. (available on-line at arXiv:1308.2462)


   %% Random beamforming papers
   \bibitem{Viswanath}
      P. Viswanath, D. N. C. Tse, and R. Laroia, ``Opportunistic beamforming using dumb antennas,'' \emph{IEEE Trans. Inf. Theory}, vol. 48, no. 6, pp. 1277-1294, June 2002

   \bibitem{Sharif}
      M. Sharif and B. Hassibi, ``On the capacity of MIMO broadcast channels with partial side information,'' \emph{IEEE Trans. Inf. Theory}, vol. 51, no. 2, pp. 506-522, Feb. 2005.

   \bibitem{Alon}
      N. Alon, J. Edmonds, and M. Luby, ``Linear time erasure codes with nearly optimal recovery,'' in \emph{Proc. the 36th IEEE Annual Symp. Foundations of Computer Science}, pp. 512-519, Oct. 1995.

   \bibitem{Alamouti}
      S. M. Alamouti, ``A simple transmit diversity technique for wireless communications,'' \emph{IEEE J. Sel. Areas Commun.}, vol. 16, no. 8, pp. 1277-1294, Oct. 1998.

   \bibitem{Proakis}
      J. G. Proakis, Digital Communication, 4th edition, McGraw-Hill, 2000.

   \bibitem{TableOfIntegral}
      I. S. Gradshteyn and I. I. Ryshik, \emph{Table of Integrals, Series, and Products}, Academia Press, sixth edition, 2000.

   \bibitem{Liberti}
      J. C. Liberti and T. S. Rappaport, ``Statistics of shadowing in indoor radio channels at 900 and 1900 MHz,'' in \emph{Proc. IEEE Military Communications Conference}, vol. 3, pp. 1066-1070, Oct. 1992.

   \bibitem{Hassibi}
      B. Hassibi and T. L. Marzetta, ``Multiple-antennas and isotropically random unitary inputs: the received signal density in closed form,'' \emph{IEEE Trans. Inf. Theory}, vol. 48, no. 6, pp. 1473-1484, June 2002.

   \bibitem{Geddes}
      K. O. Geddes, M. L. Glasser, R. A. Moore, and T. C. Scott, ``Evaluation of classes of definite integrals involving elementary functions via differentiation of special functions'', Applicable Algebra in Engineering, Communication and Computing (AAECC), vol. 1, pp. 149-165, 1990.

\end{thebibliography}
\end{document}